\documentclass[runningheads]{llncs}
\usepackage{times}
\usepackage{listings}
\lstdefinelanguage{program}{keywords={let,pass,function,var,const,bool,int,void,atomic,while,do,if,then,else,assume,require,assert,call,return,rule,forall,with,new,choose,skip,task,async,yield,for,wait,type,relation,init, action, safety, invariant
  },
  morecomment=[l]{//},
  morecomment=[s]{/*}{*/},
  morecomment=[n]{(*}{*)},
  mathescape=true,
  escapeinside=`',
}
\usepackage[numbers]{natbib}
\usepackage{hyperref}
\usepackage{caption}
\usepackage{booktabs}
\usepackage{mathpartir}
\usepackage{multirow}
\usepackage{amsmath}
\usepackage{stmaryrd}
\usepackage{array}
\usepackage[T1]{fontenc}
\usepackage[nodayofweek]{datetime}
\usepackage{graphicx}
\usepackage{algorithm}
\usepackage[noend]{algpseudocode}
\usepackage[flushleft]{threeparttable}
\usepackage{subcaption}
\usepackage{wrapfig,lipsum,booktabs}

\usepackage{xcolor}
\usepackage{extarrows}
\usepackage{mathtools}
\usepackage{paralist}

\usepackage{tikz}
\usetikzlibrary{positioning,fit}

\usepackage[noabbrev,capitalise]{cleveref}

\newif\ifnitpick
\nitpickfalse

\newif\iflong
\longtrue

\newif\ifsketch
\sketchfalse

\newif\iflockserv
\lockservfalse

\newif\ifapp
\iflong
\apptrue
\else
\appfalse
\fi

\ifapp
\newcommand{\refappendix}[1]{\Cref{#1}}
\else
\newcommand{\refappendix}[1]{the extended version~\cite{extended-version}}
\fi

\crefformat{section}{\S#2#1#3} \crefformat{subsection}{\S#2#1#3}
\crefformat{subsubsection}{\S#2#1#3}

\Crefname{conjecture}{Conjecture}{Conjectures}
\Crefname{proposition}{Proposition}{Propositions}
\Crefname{lemma}{Lemma}{Lemmas}
\Crefname{corollary}{Corollary}{Corollaries}
\Crefname{example}{Example}{Examples}
\Crefname{definition}{Definition}{Definitions}

\crefname{line}{line}{lines}

\newenvironment{proof*}[1]
  {\renewcommand{\proofname}{#1}\begin{proof}}
  {\end{proof}}

\Crefname{figure}{Fig.}{Fig.}

\newcommand{\para}[1]{\vspace{2pt}\noindent\textbf{\textit{#1.}}}

\newcommand{\empheval}[1]{\textit{#1}}
\newcommand{\emphexample}[1]{\textit{#1}}

\newcommand{\ov}{\overline}

\renewcommand{\phi}{\varphi}

\newcommand{\A}{{\cal A}}

\renewcommand{\implies}{\Longrightarrow}

\newcommand{\true}{{\textit{true}}}
\newcommand{\false}{{\textit{false}}}
\newcommand{\vocabulary}{\Sigma}
\newcommand{\voc}{\vocabulary}
\newcommand{\vocdouble}{\hat{\vocabulary}}

\newcommand{\structures}{{\mbox{STRUCT}}}
\newcommand{\structOfVoc}[1]{{\structures[#1]}}

\newcommand{\struct}{{s}}
\newcommand{\Dom}{{D}}
\newcommand{\Int}{{\mathcal{I}}}
\newcommand{\diag}[1]{{\textit{Diag}(#1)}}

\newcommand{\structdom}[2]{\structOfVoc{#1}|_{#2}}

\newcommand{\form}[2]{\mathcal{F}_{#1}(#2)}
\newcommand{\vars}{V}

\newcommand{\Init}{\textit{Init}}

\newcommand{\Safety}{\mathcal{P}}  
\newcommand{\Prop}{P}
\newcommand{\Inv}{\textit{Inv}}

\newcommand{\Lset}{L}

\newcommand{\TS}{{\textit{TS}}}

\newcommand{\TR}{{\textit{TR}}}

\newcommand{\Univ}{{\forall^*}}
\newcommand{\leqLof}[1]{\sqsubseteq_{#1}}
\newcommand{\leqL}{\leqLof{\Lset}}
\newcommand{\leqU}{\leqLof{\Univ}}

\newcommand{\qoset}{\mathcal{V}}
\newcommand{\phaseauto}{\mathcal{A}}
\newcommand{\phasestruct}{\mathcal{S}}
\newcommand{\autostate}{\mathcal{Q}}
\newcommand{\autoinitset}{\mathcal{Q}_0}
\newcommand{\autoinit}{\iota}
\newcommand{\autotrans}{\mathcal{R}}
\newcommand{\edgelabel}[1]{\delta_{#1}}
\newcommand{\statelabel}[1]{\varphi_{#1}}
\newcommand{\fv}[1]{\textit{FV}({#1})}
\newcommand{\Lang}[1]{\mathcal{L}(#1)}
\newcommand{\sphase}{q}
\newcommand{\tphase}{p}

\newcommand{\UPDR}{\text{PDR}^\forall}
\newcommand{\Frame}{\mathcal{F}}

\newcommand{\Grd}{\textit{Grd}}

\newcommand{\codefont}[1]{{\footnotesize\ensuremath{\mathtt{#1}}}}
\newcommand{\lockelem}{\ell}

\newcommand{\holdslock}{\codefont{holds\_lock}}

\newcommand{\rowner}{\codefont{owner}}
\newcommand{\rtable}{\codefont{table}}
\newcommand{\rseqsent}{\codefont{seqnum\_sent}}
\newcommand{\runacked}{\codefont{unacked}}
\newcommand{\rseqrecvd}{\codefont{seqnum\_recvd}}
\newcommand{\rtransfermsg}{\codefont{transfer\_msg}}
\newcommand{\rackmsg}{\codefont{ack\_msg}}

\newcommand{\vsrc}{{\normalfont \textrm{src}}}
\newcommand{\vdst}{{\normalfont \textrm{dst}}}

\newcommand{\Tool}{{\sc mypyvy}} 
\begin{document}
\newif\ifcomments
\commentsfalse

\ifcomments
\newcommand{\sharon}[1]{{\textcolor{blue}{SS: {\em #1}}}}
\newcommand{\jrw}[1]{{\textcolor{green}{JRW: {\em #1}}}}
\newcommand{\mooly}[1]{{\textcolor{cyan}{MS: {\em #1}}}}
\newcommand{\yotam}[1]{{\textcolor{magenta}{YF: {\em #1}}}}
\newcommand{\TODO}[1]{{\textcolor{red}{TODO: {\em #1}}}}\

\else
\newcommand{\sharon}[1]{}
\newcommand{\jrw}[1]{}
\newcommand{\mooly}[1]{}
\newcommand{\yotam}[1]{}
\newcommand{\TODO}[1]{}

\fi

\newcommand{\commentout}[1]{}
\newcommand{\OMIT}[1]{}

\title{Inferring Inductive Invariants from Phase Structures}
\author{Yotam M.\ Y.\ Feldman\inst{1}\and
James R.\ Wilcox\inst{2}\and
Sharon Shoham\inst{1}\and
Mooly Sagiv\inst{1}}

\institute{Tel Aviv University, Israel \and
University of Washington, USA}
\maketitle              \begin{abstract}
Infinite-state systems such as distributed protocols are challenging to verify using interactive theorem provers or automatic verification tools.
Of these techniques, deductive verification
is highly expressive
but requires the user to annotate the system with \emph{inductive invariants}.
To relieve the user from this labor-intensive and challenging task,
\emph{invariant inference} aims to find inductive invariants automatically.
Unfortunately, when applied to infinite-state systems such as distributed protocols, existing inference techniques often diverge, which limits their applicability.

This paper proposes \emph{user-guided invariant inference} based on \emph{phase invariants}, which capture the different logical phases of the protocol.
Users conveys their intuition by specifying a \emph{phase structure}, an automaton with edges labeled by program transitions; the tool automatically infers assertions that hold in the automaton's states, resulting in a full safety proof.
The additional structure from phases guides the inference procedure towards finding an invariant.

Our results show that user guidance by phase structures facilitates successful inference beyond the state of the art.
We find that phase structures are pleasantly well matched to the intuitive reasoning routinely used by
domain experts to understand why distributed protocols are correct,
so that providing a phase structure reuses this existing intuition.
\end{abstract}
 
\section{Introduction}
Infinite-state systems such as distributed protocols remain challenging to verify despite decades of work developing interactive and automated proof techniques.
Such proofs rely on the fundamental notion of an \emph{inductive invariant}.
Unfortunately, specifying inductive invariants is difficult for users, who must often repeatedly iterate through candidate invariants before achieving an inductive invariant.
For example, the Verdi project's proof of the Raft consensus protocol used an inductive invariant with 90 conjuncts and relied on significant manual proof effort~\cite{verdi-pldi,verdi-cpp}.

The dream of \emph{invariant inference} is that users would instead be assisted by automatic procedures that could infer the required invariants.
While other domains have seen successful applications of invariant inference,
using techniques such as abstract interpretation~\cite{DBLP:conf/popl/CousotC77} and property-directed reachability~\cite{ic3,pdr},
existing inference techniques fall short for interesting distributed protocols, and often diverge while searching for an invariant.
These limitations have hindered adoption of invariant inference.

\para{Our Approach}
The idea of this paper is that invariant inference can be made drastically more effective by utilizing \emph{user-guidance} in the form of \emph{phase structures}.
We propose user-guided invariant inference, in which the user provides some additional information to guide the tool towards an invariant.
An effective guidance method must (1) match users' high-level intuition of the proof, and
(2) convey information in a way that an automatic inference tool can readily utilize to direct the search.
\ifsketch, and
(3) allow varying degrees of control over the inference procedure by providing detail in specifying the proof by the user, from a full proof to a mere sketch, so that the user will always be able to find an invariant using the tool by providing more details until the tool converges.
In this paper we introduce \emph{phase invariants} as the basis of user-guided invariant inference that enjoys these properties.
\fi
In this setting invariant inference turns a partial, high-level argument accessible to the user into a full, formal correctness proof, overcoming scenarios where procuring the proof completely automatically is unsuccessful.

Our approach places \emph{phase invariants} at the heart of both user interaction and algorithmic inference.
Phase invariants have an automaton-based form that is well-suited to the domain of distributed protocols.
They allow the user to convey a high-level temporal intuition of why the protocol is correct in the form of a \emph{phase structure}. The phase structure provides hints that direct the search and allow a more targeted generalization of states to invariants, which can facilitate inference where it is otherwise impossible.
\OMIT{
\yotam{Perhaps stop here and move to the next paragraph?}
Phase invariants capture an innate structure of protocol's correctness in a way that can be utilized to direct the search and improve the applicability of inference.
\yotam{omitting:
\yotam{consider moving these points to a discussion of phase \emph{sketches}}
Phase invariants are useful as an interface for guiding the search for an invariant, as they
(1)~are easier to specify manually than classical inductive invariants,
(2)~naturally allow presenting the correctness proof at different levels of detail, ranging
from the full proof to only the essential structure while omitting the low-level details,
(3)~capture the temporal aspects of the algorithm, and
(4)~support automatic inference for filling in missing details in the proof, guided by the automaton structure.
Our approach places phase invariants at the heart of both user interaction and algorithmic inference, leveraging these benefits.}
}

\OMIT{
\paragraph{Phase Invariants}
Phase invariants are so named because they describe executions of the protocol
as transitioning between different logical stages, or phases.
These phases form the states of an automaton, whose edges are labeled by actions of the protocol that trigger transitions between phases.
The phases themselves are labeled
by \emph{phase characterizations}: assertions that hold whenever the protocol is in that phase.

Phase invariants closely match the way domain experts already think about the correctness of distributed protocols by means of state-machine refinement \`a la \citet[e.g.][]{DBLP:books/aw/Lamport2002}.
In essence, phase invariants reveal a \emph{logical} control structure in the protocol, guided by the correctness property.
\yotam{R1: ``it is also possible that most of the benefit of the technique is simply in the structuring mechanism of phase automaton for specifying the inductive invariant. In other words, once the programmer has come up with the intuition for the phase structure and transition relations on the edges it is not much more work to write down the assertions in each state and have them be checked by standard VC checking without the more unpredictable step of invariant inference.''}

\OMIT{
Often, the latent phase structure of a protocol becomes clear only after focusing on a particular entity in the system (e.g.\ a single node or a single piece of data).
Phase invariants are thus parameterized by a \emph{view}, a finite set of elements (implicitly universally quantified)
whose evolution is described by the phase structure.
The elements in the view may be threads or processes (as is common when dealing with parameterized systems), but may also be different objects, e.g. locks, as demonstrated by our examples (see \Cref{sec:overview,sec:evaluation}).
}

\yotam{Sharon, how the technical section looks is relevant here}
A full correctness full arises from an \emph{inductive phase automaton}, in which
the phase characterizations are sufficiently strong to show that the automaton includes transitions for all actions possible at the phase,
and that the characterizations are preserved along each automaton edge.
An inductive phase invariant thus establishes a simulation relation between the protocol and the phase automaton,
parameterized by the view,
ensuring that the automaton overapproximates all executions of the protocol \jrw{merged:} for every view.

\paragraph{User-guided inference of phase invariants}
\yotam{omitting: The automaton structure of a phase invariant is more closely aligned with the intuition of protocol developers, but specifying all the details necessary for an inductive phase invariant is still challenging and tedious.
\yotam{added:} Thus in this paper we are interested in \emph{invariant inference}.}
Our approach utilizes the additional structure of phase invariants to facilitate effective \emph{invariant inference} which is \emph{guided} by the user's phase-based understanding of correctness. In our approach, users convey their intuition of the proof by specifying an automaton structure.
\ifsketch , possibly with some \emph{partial} (i.e.\ not necessarily inductive) phase characterizations.
An automaton with partial characterizations (if any) provided for the sake of inference is called a \emph{phase sketch}.
\yotam{added:} A phase sketch is still far from a full correctness proof.
\fi
Our inference procedure automatically finds phase characterizations to achieve an inductive phase invariant out of the provided phase structure.
The inference algorithm, which is parameterized by a language of possible phase characterizations,
and may return either a completed inductive phase invariant
or a witness showing that the user's phase structure is inappropriate, in the sense that no choice of phase characterizations in the given language can form an inductive phase invariant to establish safety. \yotam{Sharon, lingo also here}

From the perspective of automatic invariant inference, the given phase structure,
the localized view it induces,
and disabled transitions---actions of the system that are impossible at some phase---are utilized by our algorithm to guide the search for an inductive invariant.
The decomposition to phases guides inference towards finding an invariant.
Ultimately, the phase structure facilitates the gradual construction of a disjunctive invariant corresponding to the different phases under the scope of universal quantification over the view,
by combining manual insight with automatic inference.

\TODO{put results much earlier}

\OMIT{
Perhaps surprisingly, we show that applying the automaton's structure itself yields a speedup in inference. \TODO{move elsewhere}
One reason is that the phase structure can also specify
\emph{disabled transitions}, actions of the system that are impossible at some phase. \yotam{TODO: modify this explanation after we understand the reasons well}
Inference uses these disabled transitions as additional safety properties to guide the search.
}

\OMIT{
\jrw{cut the following paragraph?} \yotam{cutting:
From the user's perspective, where inductive invariants attempt to summarize all possible states of the system at once, phase automata instead allow the user to refer to distinct phases in the computation, each of which is summarized separately and semi-automatically. Furthermore, the witness provided in case no inductive strengthening exists may assist the user in diagnosing the failure and modifying the phase sketch in order to achieve an inductive phase invariant.
}
}

\OMIT{

From the user's perspective, phase structure express a temporal intuition of the protocol's behavior, including when certain transitions are impossible.
\ifsketch
The specification of partial characterizations builds on the local view of phases, as the provided characterizations need not hold globally but only in the specific phase.
\fi
\jrw{Okay, decided to edit and merge with previous paragraph. Is it still correct?}
Ultimately, the phase structure facilitates the gradual construction of a disjunctive invariant corresponding to the different phases under the scope of universal quantification over the view,
by combining manual insight with automatic inference.
}

\paragraph{Results}
We instantiate our user-guided phase-based inference to the language of universally quantified invariants for EPR programs,
which previous work has used to model distributed protocols~\cite{pldi/PadonMPSS16,DBLP:journals/pacmpl/PadonLSS17,DBLP:conf/pldi/TaubeLMPSSWW18}.
Using a variant of $\UPDR$~\cite{DBLP:journals/jacm/KarbyshevBIRS17}, we infer universally quantified phase characterizations, in addition to the universal quantifiers that arise from the automaton's view.
We use our approach to infer phase invariants from phase structures for several interesting distributed protocols.
We show that inference guided by a phase structure can infer proofs for distributed protocols that are currently beyond reach for inference, and can also achieve convergence faster than unguided state-of-the-art inference.

}

This paper makes the following contributions:

        (1) We present \emph{phase invariants}, an automaton-based form of safety proofs, based on the distinct logical phases of a certain view of the system.
             Phase invariants closely match the way domain experts already think about the correctness of distributed protocols by state-machine refinement \`a la Lamport~\cite[e.g.][]{DBLP:books/aw/Lamport2002}.

        (2) We describe an algorithm for inferring \emph{inductive phase invariants} from \emph{phase structures}.
              The decomposition to phases through the phase structure guides inference towards finding an invariant.
        The algorithm finds a proof over the phase structure or explains why no such proof exists. In this way, phase invariants facilitate user interaction with the algorithm. 

        (3) Our algorithm reduces the problem of inferring inductive phase invariants from phase structures to the problem of solving a linear system of Constrained Horn Clauses (CHC), irrespective of the inference technique and the logic used. In the case of universally quantified phase inductive invariants for protocols modeled in EPR (motivated by previous deductive approaches~\cite{pldi/PadonMPSS16,DBLP:journals/pacmpl/PadonLSS17,DBLP:conf/pldi/TaubeLMPSSWW18}), we show how to solve the resulting CHC using a variant of $\UPDR$~\cite{DBLP:journals/jacm/KarbyshevBIRS17}.

        (4) We apply this approach to the inference of invariants for several interesting distributed protocols. (This is the first time invariant inference is applied to distributed protocols modeled in EPR.) In the examples considered by our evaluation, transforming our high-level intuition about the protocol into a phase structure was relatively straightforward.
        The phase structures allowed our algorithm to outperform in most cases an implementation of $\UPDR$ that does not exploit such structure, facilitating invariant inference on examples beyond the state of the art and attaining faster convergence.

\iflong
It is surprising that invariant inference---operating in the realm of logical clauses and implications---can so effectively benefit from guidance by phase structures, which exhibit a much higher level of abstraction.
While there remain significant challenges to applying invariant inference on complex distributed protocols---notably, inference of invariants with quantifier alternations, necessary, e.g.\ for Paxos~\cite{DBLP:journals/pacmpl/PadonLSS17}---our approach demonstrates that the seemingly inherent intractability of sifting through a vast space of candidate invariants can be mitigated by leveraging users' high-level intuition.
\else
Overall, our approach demonstrates that the seemingly inherent intractability of sifting through a vast space of candidate invariants can be mitigated by leveraging users' high-level intuition. \OMIT{make sure to mention Paxos somewhere}
\fi

\OMIT{We should find a natural place somewhere in the intro to mention Ivy in passing and cite,
  so that the overview does not have to explain what Ivy is.}

\OMIT{
We seek to improve on these results by leveraging two further types of automation.
First, we model systems using decidable fragments of first-order logic, leading to fully automated checking of inductive invariants.
Second, we also semi-automate the process finding a state-machine abstraction of the protocol, by allowing the user to sketch a \emph{phase automaton}.
The details of the abstraction are then filled out automatically by solving a system of second-order constrained Horn clauses using a variant of $\UPDR$.
\jrw{I was just trying to sound fancy with the previous sentence. Please check that I used all the words properly.}

Building on existing literature from the distributed systems community, our key insight is that state-machine refinement techniques such as our phase automata
more naturally capture the structure of distributed protocols.
Where inductive invariants attempt to summarize all possible states of the system at once, phase automata instead allow the user to refer to distinct phases in the computation, each of which is summarized separately and semi-automatically.
Such a summary describes the mechanisms used by the protocol during execution, which provide the essential reason for its correctness to be checked in the proof.

\TODO{explain what is a phase invariant}

We demonstrate the benefits of phase invariants by using them to verify to sophisticated protocol from the distributed systems literature.
We also compare our phase automaton characterizations to traditional inductive invariants used in state-of-the-art deductive verification tools.
We find that phase invariants improve proof automation by localizing the search, yielding a speedup of \TODO{insert results} over baseline inductive invariant inference.
Phase invariants also facilitate user interaction since users can provide hints about the meaning of each phase, which the system either completes to a proof or explains why no such proof exists.
When given such guidance, our results further improve to \TODO{insert results}, showing that users can trade off machine time for their own effort.

}

\OMIT{
    \yotam{The use of phase invariants for infinite-state systems relies on the use of a
\emph{view}. The view is a finite list of elements, fixed throughout the execution, to which the phases refer. Intuitively, each set of elements from the domain induce a phase automaton, and they all execute in parallel. The proofs we consider are \emph{view-modular}, meaning the proof of a phase invariant refers only to single arbitrary choice of the elements for the view, and not to the automata of other choices of view. \TODO{probably not clear at all...} This is sound because the view is in essence universally quantified.
One observation of this work is that often the safety property calls for a view on elements other than nodes or threads. For example, the view can consist of a shared resource or two different values in a consensus protocol.
}
}

\OMIT{

\paragraph{Inferring phase invariants}
We note that providing fully inductive (phase) invariants in distributed systems can be tricky. Typically, the user has to account for the exact inductive argument for correctness, including all corner cases. Therefore, we also describe and implement an algorithm for inferring phase invariants. This is an adaptation of IC3/PDR~\cite{} and $\UPDR$~\cite{} for inferring quantified invariants.
\TODO{insert results here}

\paragraph{Why is deductive verification hard?}
\paragraph{Why is deductive verification of distributed protocols hard?}
-> Inductive invariants are properties of states rather than traces.
More restricted than temporal logic. (Is it contained in LTL?)

\paragraph{Why is inference hard?}
\begin{itemize}
    \item CNF good for PDR, DNF for AI (?). Natural invariants of distributed protocols are a disjunction of clauses, like we find here
    \item Problem of disjunction
\end{itemize}

Phases:
Better fit.
Easy to partially characterize phases, to be completed by inference.
Even just the structure provides valuable insight that can improve inference.

\subsection{Some Potential Phrases}
\begin{itemize}
    \item Distributed protocols are often viewed and described as progressing between distinct phases of the computation~\cite{}.
    Our key insight is that \emph{correctness proofs} of distributed protocols enjoy a similar structure.

    \item Providing control to unstructured protocols

    \item While an inductive phase invariant induces a standard inductive invariant, we view the phase invariant as the fundamental building block of the proof, as it more closely captures the intuition users have of the protocol's correctness. \yotam{This one is even more dubious than the rest :)}
\sharon{maybe the point is that if we want to convert it to CNF it will become exponentially larger and much less clear?}

    \item When the phase structure of the proof deviates from that of the protocol, this elucidates the subtlety in the proof in the form of a refinement of the protocol's straightforward, operational, description.
\sharon{Here I assume you refer to the one round vs two rounds in paxos, but I am not sure that this is a good thing. }

    \item For the purpose of automatic verification, we harness the phase structure in a PDR-based invariant inference procedure, which aims to complete a self-sufficient phase invariant out of a phase sketch.
\sharon{more generally: given a phase sketch, we formulate/reduce the problem of computing a phase invariant using CHCs. Talk about linear vs. non linear?}

    \item Impossible transitions are valuable in directing the search for invariant, as they constitute additional safety constraints to be satisfied by phase characterizations.
\sharon{vague thought: can the guards on transitions be used to give us the ``bounded occurrences of arbitrary relations'' that were used in POPL16 for decidability?}

    \item (Phase specification can be useful also bug finding, restricting the possible transitions in phases)

    \item Making the distinction between phases is useful when there is a property which holds in one phase but not in the other, and when this fact is important to the proof.
    Without separating the phases, the invariant would guard the property by some other formula expressing indirectly which phase the execution is currently in (as the lockservice example does).

    \item The phases are parameterized by a certain view of the system, e.g., the view of a single thread/round or the view of a pair of threads/rounds.
    \sharon{this is related to the point about control: in distributed protocols there is no clear control partly because the players are not fixed, but when we fix a view, the control does become clear.}

    \item (?) Phase invariants combine intuition from abstraction (simulation as in existential abstraction) and deductive verification (inductiveness).

    \item (?) Phase invariants provide a convenient way for a user to convey intuition to a tool, with a spectrum of how much information can be provided, and with "focused" counterexamples from the tool (since they refer to specific phases).

\end{itemize}
}  
\iflong
\else
\vspace{-0.4cm}
\fi
\section{Preliminaries}
\iflong
\else
\vspace{-0.4cm}
\fi
\iflong
In this section we provide background on modeling and verifying systems using first-order logic.
\else
In this section we provide background on first-order transition systems.
\fi
\iflong
We assume familiarity with first-order logic.
\fi \iflong
We use many-sorted first order logic, but omit sorts here to simplify the presentation.
\else
Sorts are omitted for simplicity.
\fi
\iflong
Although in this paper we will mostly deal with uninterpreted first-order logic, our definitions and results extend to logics with a background theory.
\else
Our results extend also to logics with a background theory.
\fi

\para{Notation}
$\fv{\varphi}$ denotes the set of free variables of $\varphi$.
$\form{\voc}{\vars}$ denotes the set of first-order \iflong well-formed \fi formulas \iflong $\varphi$\fi over vocabulary $\voc$ with $\fv{\varphi} \subseteq \vars$.
We
\iflong
extend the notation $\implies$ to quantified formulas and
\fi
write $\forall \qoset.\ \varphi \implies \psi$ to denote that the formula $\forall \qoset.\ \varphi \to \psi$ is valid.
We sometimes use $f_a$ as a shorthand for $f(a)$.

\para{Transition systems}
We represent transition systems symbolically, via formulas in first-order logic. The definitions are standard.
A vocabulary $\vocabulary$ consisting of constant, function, and relation symbols is used to represent states\iflong (each function and relation symbol is associated with its arity)\fi.
Post-states of transitions are represented by a copy of $\vocabulary$ denoted $\vocabulary' = \{ a' \mid a \in \vocabulary\}$\iflong (where the arity of each function and relation symbol is inherited from $\vocabulary$)\fi.
A \emph{first-order transition system} over \iflong vocabulary\fi $\vocabulary$ is a tuple $\TS = (\Init, \TR)$, where $\Init \in \form{\voc}{\emptyset}$ describes the initial states, and $\TR \in \form{\vocdouble}{\emptyset}$ with $\vocdouble = \vocabulary \uplus \vocabulary'$ describes the transition relation.The states of $\TS$ are first-order structures over $\voc$\iflong, denoted $\structOfVoc{\voc}$\fi.
\iflong
Each state $\struct \in \structOfVoc{\voc}$ is a pair $\struct = (\Dom,\Int)$ where $\Dom$ is the \emph{domain} and $\Int$ is the \emph{interpretation function} mapping each symbol in $\voc$ to its interpretation over $\Dom$. We denote by $\structdom{\voc}{\Dom}$ the set of structures with domain $\Dom$.
In this way, every closed formula over $\voc$ represents the set of states (first-order structures) that satisfy it.
\fi
\iflong
In particular, a state $\struct$ is initial if $\struct \models \Init$. A transition of $\TS$ is a pair of states $\struct_1 = (\Dom, \Int_1),\struct_2 = (\Dom, \Int_2)$ with a shared domain such that $(\struct_1,\struct_2) \models \TR$, where $(\struct_1,\struct_2)$ is a shorthand for the structure $\struct = (\Dom, \Int)$ obtained by defining $\Int(a) = \Int_1(a)$ if $a \in \voc$, and $\Int(a) = \Int_2(a)$ if $a \in \voc'$.
\else
A state $\struct$ is initial if $\struct \models \Init$. A transition of $\TS$ is a pair of states $\struct_1,\struct_2$ over a shared domain such that $(\struct_1,\struct_2) \models \TR$, $(s_1,s_2)$ being the structure over that domain in which $\voc$ in interpreted as in $\struct_1$ and $\voc'$ as in $\struct_2$.
\fi
$\struct_1$ is also called the \emph{pre-state} and $\struct_2$ the \emph{post-state}.
Traces are finite
sequences of states $\sigma_1, \sigma_2, \ldots$ starting from an initial state
such that there is a transition between each pair of consecutive states.
The \emph{reachable states} \iflong of $\TS$\fi are those that reside on traces starting from an initial state.

\para{Safety}
\iflong
A safety property $\Prop$ is a formula in $\form{\voc}{\emptyset}$. We say that $\TS$ is \emph{safe} if all the reachable states satisfy $\Prop$, in which case we also say that $\Prop$ is an \emph{invariant} of $\TS$.
A prominent way to prove safety is via \emph{inductive invariants}. An inductive invariant $\Inv$ is a closed first-order formula over $\voc$ such that the following requirements hold:
\begin{inparaenum}[(i)]
\item $\Init \implies \Inv$ (initiation), and
\item $\Init \wedge \TR \implies \Inv'$ (consecution), \end{inparaenum}
where $\Inv'$ is obtained from $\Inv$ by replacing each symbol from $\voc$ with its primed counterpart.Initiation and consecution ensure that all the reachable states satisfy $\Inv$. If, in addition, $\Inv$ satisfies:
\begin{inparaenum}[(iii)]
\item $\Inv  \implies \Prop$ (safety),
\end{inparaenum}
it follows that all the reachable states satisfy $\Prop$, and
$\TS$ is safe.
\else
A safety property $\Prop$ is a formula in $\form{\voc}{\emptyset}$. We say that $\TS$ is \emph{safe}, and that $\Prop$ is an \emph{invariant}, if all the reachable states satisfy $\Prop$.
$\Inv \in \form{\voc}{\emptyset}$ is an \emph{inductive invariant} if
\begin{inparaenum}[(i)]
\item $\Init \implies \Inv$ (initiation), and
\item $\Init \wedge \TR \implies \Inv'$ (consecution),
\end{inparaenum}
where $\Inv'$ is obtained from $\Inv$ by replacing each symbol from $\voc$ with its primed counterpart.
If also
\begin{inparaenum}[(iii)]
\item $\Inv  \implies \Prop$ (safety),
\end{inparaenum}
then it follows that
$\TS$ is safe.
\fi  
\iflockserv
\input{running-lockserv}
\else
\section{Running Example: Distributed Key-Value Store}
\label{sec:running}
\iflong

\begin{figure*}[t]
  \centering
\begin{minipage}{1.1\textwidth}
\begin{minipage}{0.53\textwidth}
  \begin{lstlisting}[numbers=left, numberstyle=\tiny, numbersep=5pt, name=lockserv, escapeinside={(*}{*)}]
type key(*\label{kv-code:types-start}*)
type value
type node
type sequnum(*\label{kv-code:types-end}*)

relation owner: node, key(*\label{kv-code:state-start}*)
relation table: node, key, value
relation transfer_msg: node, node, 
                       key, value, seqnum
relation ack_msg: node, node, seqnum
relation seqnum_sent: node, seqnum
relation unacked: node, node, 
                  key, value, seqnum
relation seqnum_recvd: node, node, seqnum(*\label{kv-code:state-end}*)

init $\forall n_1,n_2,k.$(*\label{kv-code:init-start}*) owner($n_1$,$k$)$\land$owner($n_2$,$k$)
                  $\rightarrow n_1=n_2$(*\label{kv-code:init-start}*)
init // all other relations are empty(*\label{kv-code:init-end}*)

action reshard(n_old:node, n_new:node,(*\label{kv-code:action-start}*)(*\label{kv-code:reshard-start}*)
               k:key, value:sequnum)
  require table(n_old, k, v)
        $\land$$\neg$seqnum_sent(n_old, s)
  seqnum_sent(n_old, s) := true
  table(n_old, k, v) := false
  owner(n_old, k) := false(*\label{kv-code:cede-own}*)
  transfer_msg(n_old, n_new, k, v, s) := true(*\label{kv-code:inserting-tuple}*)
  unacked(n_old, n_new, k, v, s) := true(*\label{kv-code:reshard-end}*)

action drop_transfer_msg(src:node, dst:node,(*\label{kv-code:transfer-drop-start}*)
                    k:key, v:value, s:seqnum)
    require transfer_msg(src, dst, k, v, s)(*\label{kv-code:requiring-tuple}*)
    transfer_msg(src, dst, k, v, s) := false(*\label{kv-code:removing-tuple}*)(*\label{kv-code:transfer-drop-end}*)

action retransmit(src:node, dst:node,(*\label{kv-code:transfer-retransmit-start}*)
                  k:key, v:value, s:seqnum)
    require unacked(src, dst, k, v, s)
    transfer_msg(src, dst, k, v, s) := true(*\label{kv-code:transfer-retransmit-end}*)
  \end{lstlisting}
\end{minipage}
\hspace{0.0\textwidth}\begin{minipage}{0.4\textwidth}
  \begin{lstlisting}[numbers=left, numberstyle=\tiny, numbersep=5pt, name=lockserv, escapeinside={(*}{*)}]
action recv_transfer_msg(src:node, n:node,(*\label{kv-code:receive-transfer-start}*)
                         k:key, v:value, s:seqnum)
    require transfer_msg(src, n, k, v, s)
          $\land$$\neg$seqnum_recvd(n, src, s)(*\label{kv-code:ignore-dup}*)
    seqnum_recvd(n, src, s) := true
    table(n, k, v) := true
    owner(n, k) := true(*\label{kv-code:receive-transfer-end}*)

action send_ack(src:node, n:node,(*\label{kv-code:ack-start}*)
                k:key, v:value, s:seqnum)
    require transfer_msg(src, n, k, v, s)
          $\land$seqnum_recvd(n, src, s)
    ack_msg(src, n, s) := true(*\label{kv-code:ack-end}*)

action drop_ack_msg(src:node, dst:node,(*\label{kv-code:ack-drop-start}*)
                    k:key, s:seqnum)
    require ack_msg(src, dst, s)
    ack_msg(src, dst, s) := false(*\label{kv-code:ack-drop-end}*)

action recv_ack_msg(src:node, dst:node, 
                    k:key, s:seqnum)
    require ack_msg(src, dst, s)
    unacked(src, dst, *, *, s) := false(*\label{kv-code:qf-update:recv_ack_msg}*)

action put(n:node, k:key, v:value)
    require owner(n, k)
    table(n, k, *) := false(*\label{kv-code:qf-update:put}*)
    table(n, k, v) := true(*\label{kv-code:action-end}*)

safety $\forall k, n_1, n_2, v_1, v_2.$(*\label{kv-code:safety-start}*)
  table($n_1$,$k$,$v_1$) $\land$
  table($n_2$,$k$,$v_2$) $\to$
  $n_1 = n_2 \land v_1 = v_2$(*\label{kv-code:safety-end}*)
  \end{lstlisting}
\end{minipage}
\captionof{figure}{\footnotesize Sharded key-value store with retransmissions (KV-R) in a first-order relational modeling.}
  \label{fig:kv-ivy}
  \end{minipage}
\end{figure*} 

 \else

\begin{figure*}[t]
  \centering
\begin{minipage}{1.1\textwidth}
\begin{minipage}{0.53\textwidth}
  \begin{lstlisting}[numbers=left, numberstyle=\tiny, numbersep=5pt, name=lockserv, escapeinside={(*}{*)}]
type key(*\label{kv-code:types-start}*)
type value
type node
type sequnum(*\label{kv-code:types-end}*)

relation owner: node, key(*\label{kv-code:state-start}*)
relation table: node, key, value
relation transfer_msg: node, node, 
                       key, value, seqnum
relation ack_msg: node, node, seqnum
relation seqnum_sent: node, seqnum
relation unacked: node, node, 
                  key, value, seqnum
relation seqnum_recvd: node, node, seqnum(*\label{kv-code:state-end}*)

init $\forall n_1,n_2,k.$(*\label{kv-code:init-start}*) owner($n_1$,$k$)$\land$owner($n_2$,$k$)
                  $\rightarrow n_1=n_2$(*\label{kv-code:init-start}*)
init // all other relations are empty(*\label{kv-code:init-end}*)

action reshard(n_old:node, n_new:node,(*\label{kv-code:action-start}*)(*\label{kv-code:reshard-start}*)
               k:key, value:sequnum)
  require table(n_old, k, v)
        $\land$$\neg$seqnum_sent(n_old, s)
  seqnum_sent(n_old, s) := true
  table(n_old, k, v) := false
  owner(n_old, k) := false(*\label{kv-code:cede-own}*)
  transfer_msg(n_old, n_new, k, v, s) := true(*\label{kv-code:inserting-tuple}*)
  unacked(n_old, n_new, k, v, s) := true(*\label{kv-code:reshard-end}*)

action drop_transfer_msg(src:node, dst:node,(*\label{kv-code:transfer-drop-start}*)
                    k:key, v:value, s:seqnum)
    require transfer_msg(src, dst, k, v, s)(*\label{kv-code:requiring-tuple}*)
    transfer_msg(src, dst, k, v, s) := false(*\label{kv-code:removing-tuple}*)(*\label{kv-code:transfer-drop-end}*)

action retransmit(src:node, dst:node,(*\label{kv-code:transfer-retransmit-start}*)
                  k:key, v:value, s:seqnum)
    require unacked(src, dst, k, v, s)
    transfer_msg(src, dst, k, v, s) := true(*\label{kv-code:transfer-retransmit-end}*)
  \end{lstlisting}
\end{minipage}
\hspace{0.0\textwidth}\begin{minipage}{0.4\textwidth}
  \begin{lstlisting}[numbers=left, numberstyle=\tiny, numbersep=5pt, name=lockserv, escapeinside={(*}{*)}]
action recv_transfer_msg(src:node, n:node,(*\label{kv-code:receive-transfer-start}*)
              k:key, v:value, s:seqnum)
    require transfer_msg(src, n, k, v, s)
          $\land$$\neg$seqnum_recvd(n, src, s)(*\label{kv-code:ignore-dup}*)
    seqnum_recvd(n, src, s) := true
    table(n, k, v) := true
    owner(n, k) := true(*\label{kv-code:receive-transfer-end}*)

action send_ack(src:node, n:node,(*\label{kv-code:ack-start}*)
              k:key, v:value, s:seqnum)
    require transfer_msg(src, n, k, v, s)
          $\land$seqnum_recvd(n, src, s)
    ack_msg(src, n, s) := true(*\label{kv-code:ack-end}*)

action drop_ack_msg(src:node, dst:node,(*\label{kv-code:ack-drop-start}*)
              k:key, s:seqnum)
    require ack_msg(src, dst, s)
    ack_msg(src, dst, s) := false(*\label{kv-code:ack-drop-end}*)

action recv_ack_msg(src:node, dst:node, 
              k:key, s:seqnum)
    require ack_msg(src, dst, s)
    unacked(src, dst, *, *, s) := false(*\label{kv-code:qf-update:recv_ack_msg}*)

action put(n:node, k:key, v:value)
    require owner(n, k)
    table(n, k, *) := false(*\label{kv-code:qf-update:put}*)
    table(n, k, v) := true(*\label{kv-code:action-end}*)

safety $\forall k, n_1, n_2, v_1, v_2.$(*\label{kv-code:safety-start}*)
  table($n_1$,$k$,$v_1$) $\land$
  table($n_2$,$k$,$v_2$) $\to$
  $n_1 = n_2 \land v_1 = v_2$(*\label{kv-code:safety-end}*)
  \end{lstlisting}
\end{minipage}
\captionof{figure}{\footnotesize Sharded key-value store with retransmissions (KV-R) in a first-order relational modeling.}
  \label{fig:kv-ivy}
  \end{minipage}
\end{figure*} 

 \fi

We begin with a description of the running example we refer to throughout the paper.

The \emph{sharded key-value store with retransmissions (KV-R)}, adapted
from IronFleet~\cite[\S 5.2.1]{IronFleet}, is a distributed hash table where each node owns a subset of the keys, and keys can be dynamically transferred among nodes to balance load.
 The safety property ensures that each key is globally associated with one value, even in the presence of key transfers.
 Messages might be dropped by the network, and the protocol uses retransmissions and sequence numbers to maintain availability and safety.

\Cref{fig:kv-ivy} shows code modeling the protocol in a relational first-order language akin to Ivy~\cite{DBLP:conf/sas/McMillanP18}, which compiles to EPR transition systems.
The state of nodes and the network is modeled by global relations.
\Crefrange{kv-code:types-start}{kv-code:types-end} declare uninterpreted
sorts for keys, values, clients, and sequence numbers.
\Crefrange{kv-code:state-start}{kv-code:state-end} describe the
state, consisting of:
\begin{inparaenum}[(i)]
	\item local state of clients pertaining to the table (which nodes are \rowner{}s of which keys, and the local shard of the \rtable{} mapping keys to values);
	\item local state of clients pertaining to sent and received messages (\rseqsent{}, \runacked{}, \rseqrecvd{}); and
	\item the state of the network, comprised of two kinds of messages (\rtransfermsg{}, \rackmsg{}).
	Each message kind is modeled as a relation whose first two arguments indicate the source and destination of the message, and the rest carry the message's payload.
	For example, \rackmsg{} is a relation over two nodes and a sequence number, with the intended meaning that a tuple $(c_1, c_2, s)$ is in
	\rackmsg{} exactly when there is a message in the network from $c_1$ to $c_2$ acknowledging a message with sequence number $s$.
\end{inparaenum}

The initial states are specified in \crefrange{kv-code:init-start}{kv-code:init-end}.
Transitions are specified by the actions declared in \crefrange{kv-code:action-start}{kv-code:action-end}.
Actions can fire nondeterministically at any time when their precondition (\codefont{require} statements) holds.
Hence, the transition relation comprises of the disjunction of the transition relations induced by the actions.
The state is mutated by modifying the relations. For example,
message sends are modeled by inserting a tuple into the corresponding relation
(e.g. \cref{kv-code:inserting-tuple}),
while message receives are modeled by requiring a tuple to be in the relation
(e.g. \cref{kv-code:requiring-tuple}),
and then removing it
(e.g. \cref{kv-code:removing-tuple}).
The updates in \cref{kv-code:qf-update:recv_ack_msg,kv-code:qf-update:put} remove a set of tuples matching the pattern.

\iflong
\para{KV-R protocol}
\fi
Transferring keys between nodes begins by sending a \rtransfermsg{} from the owner to a new node (\cref{kv-code:reshard-start}), which stores the key-value pair when it receives the message (\cref{kv-code:receive-transfer-start}).
Upon sending a transfer message the original node cedes ownership (\cref{kv-code:cede-own}) and does not send new transfer messages.
Transfer messages may be dropped (\cref{kv-code:transfer-drop-start}). To ensure that the key-value pair is not lost, retransmissions are performed (\cref{kv-code:transfer-retransmit-start}) with the same sequence number until the target node acknowledges (which occurs in \cref{kv-code:ack-start}). Acknowledge messages themselves may be dropped (\cref{kv-code:ack-drop-start}).
Sequence numbers protect from delayed transfer messages, which might contain old values (\cref{kv-code:ignore-dup}).

\iflong
\para{KV-R safety property}
\fi
\Crefrange{kv-code:safety-start}{kv-code:safety-end}
specify the key safety property: at most one value is associated with any key, anywhere in the network.
Intuitively, the protocol satisfies this because each key $k$ is either currently
\begin{inparaenum}[(1)]
	\item \emph{owned} by a node, in which case this node is unique, or
	\item it is in the process of \emph{transferring} between nodes, in which case the careful use of sequence numbers ensures that the destination of the key is unique.
\end{inparaenum}
As is typical, it is not straightforward to translate this intuition into a full correctness proof. In particular, it is necessary to relate all the different components of the state, including clients' local state and pending messages.
\OMIT{This is a more elaborate version that makes it seem easy to write by hand:
When $k$ is \emph{owned}, it is owned by a single node, and there are no pending messages concerning this key. When $k$ is transferring \emph{transferring}, the key is not owned by any node, and pending non-duplicate transfer messages are unique.
}
\OMIT{In this flow I don't have a version of this:
As is typical, the lock service safety property is not inductive by itself,
because it is not strong enough to be preserved by the transitions of the system.
In particular, since the safety property only mentions \holdslock{}, it does not
constrain the other components of the state (including all network messages and the server state).}

\emph{Invariant inference} strives to automatically find an inductive invariant establishing safety.
This example is challenging for existing inference techniques (\Cref{sec:evaluation}).
This paper proposes \emph{user-guided invariant inference} based on \emph{phase-invariants} to overcome this challenge. The rest of the paper describes our approach, in which inference is provided with the phase structure in
\iflockserv
 \Cref{fig:lockserv-automaton},
 \else
 \Cref{fig:kv-automaton},
 \fi
 matching the high level intuitive explanation above. The algorithm then automatically infers facts about each phase to obtain an inductive invariant. \Cref{sec:phase-invariants} describes phase structures and inductive phase invariants, and \Cref{sec:inference} explains how these are used in user-guided invariant inference.

\commentout{
\subsection{Inductive Invariants for Distributed Protocols}
The standard approach to formally prove safety for a system such as the lock service is to find an
\emph{inductive invariant}, which is a property that
(1) implies safety,
(2) is true in the initial state, and
(3) is preserved by all transitions of the system.
}

\commentout{
	Unfortunately, inductive invariants are notoriously hard and laborious to specify, making \emph{invariant inference} an appealing direction.
In this setting an automatic algorithm finds an inductive invariant for a given system and safety property, relieving the user from specifying it by hand.
Unfortunately, state-of-the-art invariant inference techniques typically fall short for interesting distributed protocols, and often diverge in their search.
\OMIT{
\yotam{Include only if we can do better:
Furthermore, these techniques suffer from sensitivity to minor details in the program---for example, insubstantial details in the data representation might throw them off course---and even to the details of the underlying solver they use.}
These problems hinder the applicability of invariant inference techniques to distributed protocols.
}	
}

\commentout{
\paragraph{Inference implementation}
Our algorithm reduces the inference of inductive phase invariants from a given phase structure to a set of linear Constrained Horn Clauses.
Such systems can be solved in various ways, including abstract interpretation and property-directed reachability (PDR). We implement a PDR-based inference procedure for universally quantified characterizations over uninterpreted first-order logic, based on $\UPDR$~\cite{DBLP:journals/jacm/KarbyshevBIRS17}, and apply it to infer inductive phase invariants for various distributed protocols modeled in EPR. The implementation and evaluation are described in \Cref{sec:evaluation}.	
}   
\section{Phase Structures and Invariants}
\label{sec:phase-invariants}
\iflong
Phase invariants describe the
protocol as transitioning
between different logical stages.
\fi
In this section we introduce \emph{phase structures} and \emph{inductive phase invariants}\iflong
and explain their role in verifying safety properties.
\else. \fi
\iflong
In \Cref{sec:inference} we explain how we use these in guiding automatic invariant inference.
\else
These are used for guiding automatic invariant inference in \Cref{sec:inference}.
\fi
\ifapp
For brevity, proofs are deferred to \refappendix{sec:proofs}.
\else
Proofs appear in~\cite{extendedVersion}. \fi

\subsection{Phase Invariants}
\begin{definition}[Quantified Phase Automaton]
A \emph{quantified phase automaton} (\emph{phase automaton} for short) over $\voc$ is a tuple
$\phaseauto = (\autostate, \autoinit, \qoset, \edgelabel{}, \statelabel{})$ where:
\iflong
\begin{itemize}
	\item $\autostate$ is a finite set of \emph{phases}.
\item $\autoinit \in \autostate$ is the initial phase.
	\item $\qoset$ is a set of variables, called the \emph{automaton's quantifiers} \sharon{or view, or view quantifiers?}.
	\item $\edgelabel{}: \autostate \times \autostate \to \form{\vocdouble}{\qoset}$ is a function labeling every pair of phases by a transition relation formula, such that $\fv{\edgelabel{(\sphase,\tphase)}} \subseteq \qoset$ for every $(\sphase,\tphase) \in \autostate \times \autostate$. \item $\statelabel{}: \autostate \to \form{\voc}{\qoset}$ is a function labeling every phase by a \emph{phase characterization} formula, such that $\fv{\statelabel{\sphase}} \subseteq \qoset$ for every phase $\sphase \in \autostate$.
\end{itemize}
\else
$\autostate$ is a finite set of \emph{phases}.
$\autoinit \in \autostate$ is the initial phase.
$\qoset$ is a set of variables, called the \emph{automaton's quantifiers}.
$\edgelabel{}: \autostate \times \autostate \to \form{\vocdouble}{\qoset}$ is a function labeling every pair of phases by a transition relation formula, such that $\fv{\edgelabel{(\sphase,\tphase)}} \subseteq \qoset$ for every $(\sphase,\tphase) \in \autostate \times \autostate$. $\statelabel{}: \autostate \to \form{\voc}{\qoset}$ is a function labeling every phase by a \emph{phase characterization} formula, s.t.\ $\fv{\statelabel{\sphase}} \subseteq \qoset$ for every $\sphase \in \autostate$.
\fi

\end{definition}
Intuitively, $\qoset$ should be understood as free variables that are implicitly universally quantified outside of the automaton's scope. For each assignment to these variables, the automaton represents the progress along the phases from the point of view of this assignment, and thus $\qoset$ is also called the \emph{view} (or \emph{view quantifiers}).

We refer to $(\autostate, \autoinit, \qoset, \edgelabel{})$, where $\statelabel{}$ is omitted, as the \emph{phase structure} (or the \emph{automaton structure}) of $\phaseauto$.
\iflong

\fi
We refer by the \emph{edges} of $\phaseauto$ to
$
\autotrans = \{(\sphase,\tphase) \in \autostate \times \autostate \mid \edgelabel{(\sphase,\tphase)} \not\equiv \false \}.
$
\iflong

\fi
A \emph{trace} of $\phaseauto$ is a sequence of phases $\sphase_0,\ldots,\sphase_n$ such that $q_0 = \autoinit$ and $(\sphase_i,\sphase_{i+1}) \in \autotrans$ for every $0 \leq i < n$.
\iflong

\fi
We say that $\phaseauto$ is \emph{deterministic} if for every $(\sphase,\tphase_1), (\sphase,\tphase_2) \in \autotrans$ s.t. $\tphase_1 \neq \tphase_2$,
the formula $\edgelabel{(\sphase,\tphase_1)} \wedge \edgelabel{(\sphase,\tphase_2)}$ is unsatisfiable.

\iflockserv
\input{figs/lockserv-phase-automaton-w-char}
\else
\iflong
\begin{figure*}[t]
\begin{subfigure}{\textwidth}
  \centering
\begin{tikzpicture}[
    state/.style={
      draw,
      circle,
      font=\footnotesize\ttfamily,
      minimum size=9mm,
      inner sep=0pt
    },
    action/.style={
      font=\ttfamily\scriptsize }
  ]
    \node[state] (O) at (180:1.5cm) {{\small O[$k$]}};
    \node[state] (T) at (0:1.5cm) {{\small T[$k$]}};

    \draw[->] (O) edge[bend left] node[action,above] {\footnotesize reshard(*,*,$k$,*)} (T);
\draw[->] (T) edge[bend left] node[action,below] {\footnotesize recv\_transfer\_msg(*,*,$k$,*,*)} (O);

    \draw[->] (O) edge[loop left] node[action,left,align=center] 
      {drop\_transfer\_msg(*,*,$k$,*,*) \\
       retransmit(*,*,$k$,*,*) \\
       send\_ack(*,*,$k$,*,*) \\
       drop\_ack\_msg(*,*,$k$,*) \\
       recv\_ack\_msg(*,*,$k$,*) \\
       put(*,$k$,*)} (O);
    \draw[->] (T) edge[loop right] node[action,right,align=center] 
      {drop\_transfer\_msg(*,*,$k$,*,*) \\
       retransmit(*,*,$k$,*,*) \\
       send\_ack(*,*,$k$,*,*) \\
       drop\_ack\_msg(*,*,$k$,*) \\
       recv\_ack\_msg(*,*,$k$,*) \\
       put(*,$k$,*)} (T);
  \end{tikzpicture}
\end{subfigure}
\\
\begin{subfigure}{\textwidth}
  \centering
  \begin{lstlisting}[numbers=left, numberstyle=\tiny, numbersep=5pt, name=lockserv, morekeywords={phase}]
phase O[$k$]:
  invariant $\forall n_1,n_2.$ owner($n_1$,$k$)$\land$owner($n_2$,$k$)$\rightarrow n_1=n_2$
  invariant $\forall n,v.$ table($n$,$k$,$v$)$\rightarrow$owner($n$,$k$)
  invariant $\forall \vsrc,\vdst,v,s.$ $\neg$(transfer_msg($\vsrc$,$\vdst$,$k$,$v$,$s$)$\land$$\neg$seqnum_recvd($\vdst$,$\vsrc$,$s$))$\label{kv-char:disallow-recv-transfer}$
  invariant $\forall n_1,n_2,v_1,v_2.$ table($n_1$,$k$,$v_1$)$\land$table($n_2$,$k$,$v_2$)$\rightarrow n_1 = n_2 \land v_1 = v_2$
  invariant $\forall \vsrc,\vdst,v,s.$ $\neg$(unacked($\vsrc$,$\vdst$,$k$,$v$,$s$)$\land$$\neg$seqnum_recvd($\vdst$,$\vsrc$,$s$))

phase T[$k$]:
  invariant $\forall n.$ $\neg$owner($n$,$k$)
  invariant $\forall n,v.$ table($n$,$k$,$v$)$\rightarrow$owner($n$,$k$)
  invariant $\label{kv-char:unique-transfer-unreceived}$$\forall \vsrc_1,\vsrc_2,\vdst_1,\vdst_2,v_1,v_2,s_1,s_2.$ transfer_msg($\vsrc_1$,$\vdst_1$,$k$,$v_1$,$s_1$)$\land$$\neg$seqnum_recvd($\vdst_1$,$\vsrc_1$,$s_1$)
    $\land$transfer_msg($\vsrc_2$,$\vdst_2$,$k$,$v_2$,$s_2$)$\land$$\neg$seqnum_recvd($\vdst_2$,$\vsrc_2$,$s_2$)$\rightarrow \vsrc_1 = \vsrc_2 \land \vdst_1 = \vdst_2 \land v_1=v_2 \land s_1=s_2$
  invariant $\forall \vsrc_1,\vsrc_2,\vdst_1,\vdst_2,v_1,v_2,s_1,s_2.$transfer_msg($\vsrc_1$,$\vdst_1$,$k$,$v_1$,$s_1$)$\land$$\neg$seqnum_recvd($\vdst_1$,$\vdst_1$,$s_1$)
    $\land$unacked($\vsrc_2$,$\vdst_2$,$k$,$v_2$,$s_2$)$\land$$\neg$seqnum_recvd($\vdst_2$,$\vsrc_2$,$s_2$)$\rightarrow \vsrc_1 = \vsrc_2 \land \vdst_1 \land \vdst_2 \land v_1 = v_2 \land s_1 = s_2$
  invariant $\forall \vsrc_1,\vsrc_2,\vdst_1,\vdst_2,v_1,v_2,s_1,s_2.$ unacked($\vsrc_1$,$\vdst_1$,$k$,$v_1$,$s_1$)$\land$$\neg$seqnum_recvd($\vdst_1$,$\vsrc_1$,$s_1$)
    $\land$unacked($\vsrc_2$,$\vdst_2$,$k$,$v_2$,$s_2$)$\land$$\neg$seqnum_recvd($\vdst_2$,$\vsrc_2$,$s_2$)$\rightarrow \vsrc_1 = \vsrc_2 \land \vdst_1 = \vdst_2 \land v_1 = v_2 \land s_1 = s_2$
\end{lstlisting}
\end{subfigure}
\caption{\footnotesize Phase structure for key-value store (top) and phase characterizations (bottom). The user provides the phase structure, and inference automatically produces the phase characterizations, forming a safe inductive phase automaton.
}
\label{fig:kv-automaton}
\label{fig:kv-invariants}
\label{fig:kv-phase}
\end{figure*}  \else
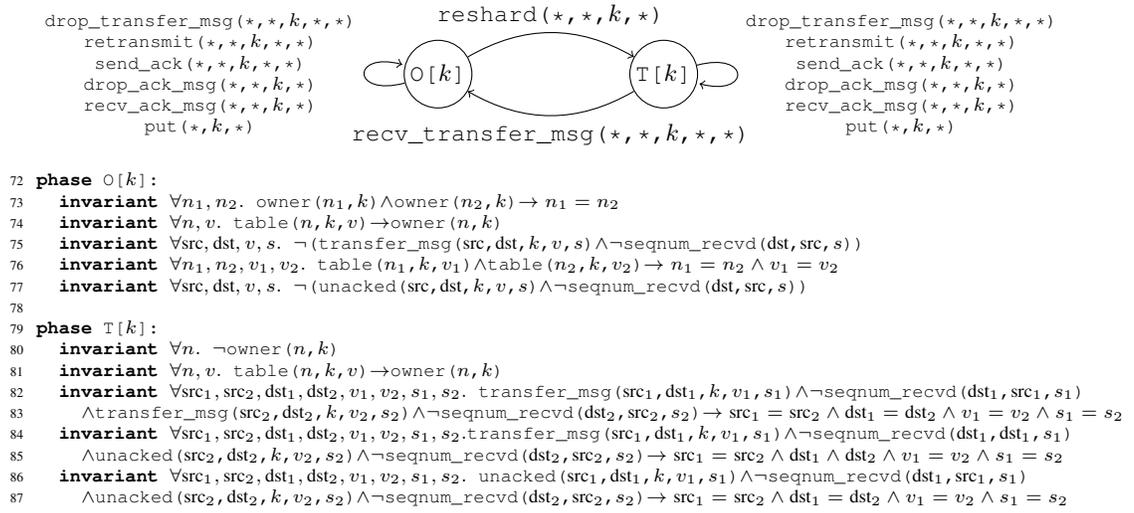
\begin{figure*}[t]
\begin{subfigure}{\textwidth}
  \centering
\begin{tikzpicture}[
    state/.style={
      draw,
      circle,
      font=\footnotesize\ttfamily,
      minimum size=9mm,
      inner sep=0pt
    },
    action/.style={
      font=\ttfamily\scriptsize }
  ]
    \node[state] (O) at (180:1.5cm) {{\small O[$k$]}};
    \node[state] (T) at (0:1.5cm) {{\small T[$k$]}};

    \draw[->] (O) edge[bend left] node[action,above] {\footnotesize reshard(*,*,$k$,*)} (T);
\draw[->] (T) edge[bend left] node[action,below] {\footnotesize recv\_transfer\_msg(*,*,$k$,*,*)} (O);

    \draw[->] (O) edge[loop left] node[action,left,align=center] 
      {drop\_transfer\_msg(*,*,$k$,*,*) \\
       retransmit(*,*,$k$,*,*) \\
       send\_ack(*,*,$k$,*,*) \\
       drop\_ack\_msg(*,*,$k$,*) \\
       recv\_ack\_msg(*,*,$k$,*) \\
       put(*,$k$,*)} (O);
    \draw[->] (T) edge[loop right] node[action,right,align=center] 
      {drop\_transfer\_msg(*,*,$k$,*,*) \\
       retransmit(*,*,$k$,*,*) \\
       send\_ack(*,*,$k$,*,*) \\
       drop\_ack\_msg(*,*,$k$,*) \\
       recv\_ack\_msg(*,*,$k$,*) \\
       put(*,$k$,*)} (T);
  \end{tikzpicture}
\end{subfigure}
\\
\begin{subfigure}{\textwidth}
  \centering
  \begin{lstlisting}[numbers=left, numberstyle=\tiny, numbersep=5pt, name=lockserv, morekeywords={phase}]
phase O[$k$]:
  invariant $\forall n_1,n_2.$ owner($n_1$,$k$)$\land$owner($n_2$,$k$)$\rightarrow n_1=n_2$
  invariant $\forall n,v.$ table($n$,$k$,$v$)$\rightarrow$owner($n$,$k$)
  invariant $\forall \vsrc,\vdst,v,s.$ $\neg$(transfer_msg($\vsrc$,$\vdst$,$k$,$v$,$s$)$\land$$\neg$seqnum_recvd($\vdst$,$\vsrc$,$s$))$\label{kv-char:disallow-recv-transfer}$
  invariant $\forall n_1,n_2,v_1,v_2.$ table($n_1$,$k$,$v_1$)$\land$table($n_2$,$k$,$v_2$)$\rightarrow n_1 = n_2 \land v_1 = v_2$
  invariant $\forall \vsrc,\vdst,v,s.$ $\neg$(unacked($\vsrc$,$\vdst$,$k$,$v$,$s$)$\land$$\neg$seqnum_recvd($\vdst$,$\vsrc$,$s$))

phase T[$k$]:
  invariant $\forall n.$ $\neg$owner($n$,$k$)
  invariant $\forall n,v.$ table($n$,$k$,$v$)$\rightarrow$owner($n$,$k$)
  invariant $\label{kv-char:unique-transfer-unreceived}$$\forall \vsrc_1,\vsrc_2,\vdst_1,\vdst_2,v_1,v_2,s_1,s_2.$ transfer_msg($\vsrc_1$,$\vdst_1$,$k$,$v_1$,$s_1$)$\land$$\neg$seqnum_recvd($\vdst_1$,$\vsrc_1$,$s_1$)
    $\land$transfer_msg($\vsrc_2$,$\vdst_2$,$k$,$v_2$,$s_2$)$\land$$\neg$seqnum_recvd($\vdst_2$,$\vsrc_2$,$s_2$)$\rightarrow (\vsrc_1, \vdst_1, v_1, s_1)=(\vsrc_2, \vdst_2, v_2, s_2)$
  invariant $\forall \vsrc_1,\vsrc_2,\vdst_1,\vdst_2,v_1,v_2,s_1,s_2.$transfer_msg($\vsrc_1$,$\vdst_1$,$k$,$v_1$,$s_1$)$\land$$\neg$seqnum_recvd($\vdst_1$,$\vdst_1$,$s_1$)
    $\land$unacked($\vsrc_2$,$\vdst_2$,$k$,$v_2$,$s_2$)$\land$$\neg$seqnum_recvd($\vdst_2$,$\vsrc_2$,$s_2$)$\rightarrow (\vsrc_1, \vdst_1, v_1, s_1) = (\vsrc_2, \vdst_2, v_2, s_2)$
  invariant $\forall \vsrc_1,\vsrc_2,\vdst_1,\vdst_2,v_1,v_2,s_1,s_2.$ unacked($\vsrc_1$,$\vdst_1$,$k$,$v_1$,$s_1$)$\land$$\neg$seqnum_recvd($\vdst_1$,$\vsrc_1$,$s_1$)
    $\land$unacked($\vsrc_2$,$\vdst_2$,$k$,$v_2$,$s_2$)$\land$$\neg$seqnum_recvd($\vdst_2$,$\vsrc_2$,$s_2$)$\rightarrow (\vsrc_1, \vdst_1, v_1, s_1) = (\vsrc_2, \vdst_2, v_2, s_2)$
\end{lstlisting}
\end{subfigure}
\caption{\footnotesize Phase structure for key-value store (top) and phase characterizations (bottom). The user provides the phase structure, and inference automatically produces the phase characterizations, forming a safe inductive phase automaton.
}
\label{fig:kv-automaton}
\label{fig:kv-invariants}
\label{fig:kv-phase}
\end{figure*}  \fi
\fi

\iflockserv
\begin{example}[Quantified Phase Automaton and Structure] \label{ex:phase-structure}
\Cref{fig:lockservice-phase} shows a phase structure of the running example.
The phase invariant is structured as an automaton parameterized by a lock $\lockelem$.
The automaton, displayed in \Cref{fig:lockserv-automaton}, describes the protocol as transitioning between 4 distinct (logical) phases of $\lockelem$.
Clockwise from top, the phases for the lock $\lockelem$ are: Server Holds Lock, Grant message destined for some client, lock held by some client, and Unlock message sent from some client and pending to be received by the server.  The transitions are labeled by actions of the system; a wildcard \texttt{*} means that the action is executed by some arbitrary node.
Each phase has
a self loop for \texttt{send\_lock}, which can happen at any time. Actions related to locks other than $\lockelem$ are considered as self-loops, and omitted here for brevity. There
are four non-self-loop edges, forming the core of the automaton and
capturing the lifecycle of the lock service.
The automaton transitions between phases when the protocol executes an action that matches the automaton's edge between the phases.
For example, a \codefont{recv\_grant} for some client and the lock $\lockelem$ is possible only when starting from phase \texttt{G}($\lockelem$), in which case the automaton transitions to phase \texttt{HL}($\lockelem$); in all other phases, the transition is disallowed.
Our algorithm exploits missing edges as additional safety properties to guide the search.
Each phase in the automaton is \emph{characterized} by a set of logical formulas, depicted in \Cref{fig:lockserv-invariants}.
When these are omitted \Cref{fig:lockserv-automaton} represents a \emph{phase structure}, which is the input to our inference algorithm.
\end{example}
\else
\iflong
\begin{example}[Quantified Phase Automaton and Structure]
\else
\begin{example}
\fi
\label{ex:phase-structure}
\Cref{fig:kv-phase} shows a phase automaton for the running example, with the view of a single key $k$.
It describes the protocol as transitioning between two distinct (logical) phases of $k$: \emph{owned} ($\texttt{O}[k]$) and \emph{transferring} ($\texttt{T}[k]$).
The edges are labeled by actions of the system. A wildcard \texttt{*} means that the action is executed with an arbitrary argument.
The two central actions are (i) {\small \texttt{reshard}}, which transitions from  $\texttt{O}[k]$ to $\texttt{T}[k]$, but cannot execute in $\texttt{T}[k]$, and (ii) {\small \texttt{recv\_transfer\_message}}, which does the opposite.
The rest of the actions do not cause a phase change and appear on a self loop in each phase.
Actions related to keys other than $k$ are considered as self-loops, and omitted here for brevity.
\iflong
Intuitively, the automaton transitions between phases when the protocol executes an action that matches the automaton's edge between the phases.
\fi
Some actions are \emph{disallowed} in certain phases, namely, do not label \emph{any} outgoing edge from a phase, such as {\small \texttt{recv\_transfer\_msg($k$)}} in $\texttt{O}[k]$.
\emph{Characterizations} for each phase are depicted in \Cref{fig:kv-invariants} (bottom).
Without them, \Cref{fig:kv-automaton} represents a \emph{phase structure}, which serves as the input to our inference algorithm.
\iflong
\begin{remark}
\fi
We remark that the choice of automaton aims to reflect the safety property of interest.
\iflockserv
In the running example, one might instead imagine taking the view of a single client
as it interacts with multiple locks.
\else
In our example, one might instead imagine taking the view of a single node
as it interacts with multiple keys, which might seem intuitive
\fi
from the standpoint of implementing
the system.
However, it is not appropriate for the proof of
\iflockserv
mutual exclusion,
since locks pass in and out of the view of a single client over their lifetime.
\else
value uniqueness, since keys pass in and out of the view of a single client\iflong over their lifetime\fi.
\fi
\iflong
The phase automaton should not aim to capture the phase structure of the implementation but that arising from the correctness intuition.
\end{remark}
\fi
\fi
\end{example}
\fi

\commentout{
	\begin{remark}[Phases under a view]
	\TODO{move? explain somewhere earlier about the view?}
	The different phases are understood with respect to a \emph{view},
	\iflockserv
	a single lock in the case of the lock service.
	\else
	a single key in the case of the running example.
	\fi
	This encodes another part of the intuitive argument: the focus on a
	\iflockserv
	single lock at a time,
	which exploits the fact that different locks
	\else
	single key at a time,
	which exploits the fact that different keys
	\fi
	do not interact in this protocol. \TODO{improve} \TODO{footnote}

\end{remark}	
}

We now formally define \emph{phase invariants} as phase automata that overapproximate the behaviors of the original system.
\begin{definition}[Language of Phase Automaton]
\label{def:language-phase-automaton}
Let $\phaseauto$ be a quantified phase automaton over $\voc$, and $\ov{\sigma} = \sigma_0,\ldots,\sigma_n$ a finite sequence of states \iflong (first-order structures)\fi \iflong in $\structOfVoc{\voc}$\else over $\voc$\fi, all with domain $\Dom$.
Let $v: \qoset \to \Dom$ be a valuation of the automaton quantifiers.
We say that:
\begin{itemize}
\item $\ov{\sigma}, v \models \phaseauto$ if there exists a trace of phases $\sphase_0,\ldots,\sphase_n$ such that $(\sigma_i, \sigma_{i+1}), v \models \edgelabel{(\sphase_i, \sphase_{i+1})}$ for every $0 \leq i < n$ \iflong (in particular, $\sphase_0,\ldots,\sphase_n$ is a trace of $\phaseauto$)\fi and $\sigma_i, v \models \statelabel{\sphase_i}$ for every $0 \leq i \leq n$.

\item $\ov{\sigma} \models \phaseauto$ if $\ov{\sigma}, v \models \phaseauto$ for every valuation $v$.
\end{itemize}
The language of $\phaseauto$ is $\Lang{\phaseauto} = \{ \ov{\sigma} \mid \ov{\sigma} \models \phaseauto\}$.
\end{definition}

\iflong
\subsection{Phase Invariants}
\fi

\OMIT{omitting: We use phase automata to over-approximate the set of traces of a transition system, in which case we refer to them as phase invariants:}

\begin{definition}[Phase Invariant]
A phase automaton $\phaseauto$ is a \emph{phase invariant} for a transition system $\TS$ \iflong(both over $\voc$)\fi if $\Lang{\TS} \subseteq \Lang{\phaseauto}$, where $\Lang{\TS}$ denotes the set of \iflong all\fi finite traces of $\TS$.
\end{definition}
\commentout{
That is, $\phaseauto$ is a phase invariant of $\TS$ if the language of the latter is included in the language of the former (when viewing the set of traces of $\TS$ as its language).
In other words, $\phaseauto$ over-approximates the set of traces of $\TS$.}\iflong{
Trace inclusion ensures that a phase invariant $\phaseauto$ can be soundly used to verify safety properties of $\TS$.
\fi

\iflockserv
\iflong
\begin{example}[Phase Invariant]
\else
\begin{example}
\fi
The phase automaton of \Cref{fig:lockserv-automaton}  is a \emph{phase invariant} for the lock service protocol: the phase characterizations of each phase hold whenever an execution trace of the original system reaches the phase.
The phases and characterizations correspond exactly to the properties from the intuitive correctness argument in \Cref{sec:running}; one out of the four properties is true in each phase, while all others are false.
\end{example}
\else
\iflong
\begin{example}[Phase Invariant]
\else
\begin{example}
\fi
The phase automaton of \Cref{fig:kv-automaton} is a \emph{phase invariant} for the protocol: intuitively, whenever an execution of the protocol reaches a phase, its characterizations hold.
This fact may not be straightforward to establish. To this end we develop the notion of \emph{inductive} phase invariants.
\end{example}

\iflong
\subsection{Establishing Phase Invariants with Inductive Phase Invariants}
\else
\subsection{Establishing Safety and Phase Invariants with Inductive Phase Invariants}
\fi

To establish phase invariants, we use inductiveness:

\begin{definition}[Inductive Phase Invariant] \label{def:ind}
$\phaseauto$ is \emph{inductive w.r.t.\ $\TS = (\Init, \TR)$} if\iflong the following conditions hold\fi:
\begin{description}
\item[{\bf Initiation:}] $\Init \implies \left(\forall \qoset. \ \statelabel{\autoinit}\right)$ .
\item[{\bf Inductiveness}:] for all $(\sphase,\tphase) \in \autotrans$, \quad $\forall \qoset. \ \left(\statelabel{\sphase} \land \edgelabel{(\sphase,\tphase)} \implies \statelabel{\tphase}'\right)$.
\item[{\bf Edge Covering}:] for every $\sphase \in \autostate$, \quad $\forall \qoset. \ \left(\statelabel{\sphase} \land \TR \implies \bigvee_{(\sphase,\tphase) \in \autotrans}{\edgelabel{(\sphase,\tphase)}}\right)$ .
\end{description}
\end{definition}

\iflong
\begin{example}[Inductive Phase Invariant]
\else
\begin{example}
\fi
The phase automaton in
\iflockserv
\Cref{fig:lockservice-phase}
\else
\Cref{fig:kv-phase}
\fi
is an inductive phase invariant\iflong: the characterizations are such that in every possible phase, when the protocol executes a valid action starting from a state satisfying the current phase characterizations, there is an outgoing edge that matches this action (covering), and the resulting program state satisfies the characterization of the target phase (inductiveness)\fi.
 \iflockserv
For example, it can be seen that every transition of the protocol starting from a state that satisfies the characterizations of \texttt{G}($\lockelem$) matches the labeling of at least one automaton edge outgoing from \texttt{G}($\lockelem$), and if, for example, the transition matches the labeling of the edge to \texttt{HL}($\lockelem$) (i.e.\ a \codefont{recv\_grant}(\texttt{*},$\lockelem$) action), the protocol necessarily ends in a state satisfying the characterizations of \texttt{HL}($\lockelem$).
\else
For example, the only disallowed transition in \texttt{O[$k$]} is {\small \texttt{recv\_transfer\_message}}, which indeed cannot execute in \texttt{O[$k$]} according to the characterization in \cref{kv-char:disallow-recv-transfer}.
Further, if, for example, a protocol's transition from \texttt{O[$k$]} matches the labeling of the edge to \texttt{T[$k$]} (i.e.\ a {\small \texttt{reshard}} action on $k$), the post-state necessarily satisfies the characterizations of \texttt{T[$k$]}: for instance, the post-state satisfies the uniqueness of unreceived transfer messages (\cref{kv-char:unique-transfer-unreceived}) because in the pre-state there are none (\cref{kv-char:disallow-recv-transfer}).
\fi
\end{example}

\commentout{
The requirements listed in \Cref{eq:ind-auto-init,eq:ind-auto-consec,eq:ind-edge-covering} define \emph{verification conditions} for verifying that a given phase automaton is inductive w.r.t.\ $\TS$, \sharon{rephrase as follws: Even though the focus of this paper is on inference of inductive invariants, this already gives rise to deductive verification based on user-provided phased automata, as justified by the following lemma} and form the basis for \emph{deductive verification} based on phase automata, as justified by the following lemma.
}

\begin{lemma}
\label{lemma:phase-ind-to-phase-inv}
If $\phaseauto$ is inductive w.r.t.\ $\TS$ then it is a phase invariant for $\TS$.
\end{lemma}

\commentout{
To prove the lemma we first define a notion of \emph{simulation} which provides a sufficient condition for $\phaseauto$ being a phase invariant of $\TS$. We then show that for an inductive phase invariant, the phase characterizations induce a simulation relation.

\begin{definition}[Simulation]
\label{def:simulation}
We say that $\phaseauto$ \emph{simulates} $\TS = (\Init, \TR)$ (both over $\voc$) if for every domain $\Dom$ and every valuation $v$ over $\Dom$ there is a relation $H \subseteq \structdom{\voc}{\Dom} \times \autostate$ such that
\begin{enumerate}
	\item \label{it:simulation-labeling-agree} $(\sigma, \sphase) \in H$ implies $\sigma,v \models \statelabel{\sphase}$.
	\item \label{it:simulation-init} For every $\sigma_0 \models \Init$, $(\sigma_0, \autoinit) \in H$.
	\item \label{it:simulation-next} For every transition $(\sigma, \sigma') \in \TR$, if $(\sigma, \sphase) \in H$ then there exists $\sphase' \in \autostate$ such that  $(\sigma', \sphase') \in H$ and
	$(\sigma, \sigma'), v \models \edgelabel{(\sphase,\sphase')}$.
\end{enumerate}
\end{definition}

\begin{lemma}\label{lem:sim-to-phase-inv}
If $\phaseauto$ simulates $\TS$ then it is a phase invariant for $\TS$.
\end{lemma}
\begin{proof}
	Let $\ov{\sigma} = \sigma_0, \ldots, \sigma_n$ be a finite trace of $\TS$, and let $v$ be a valuation of $\phaseauto$'s quantifiers to the domain of $\ov{\sigma}$.
	As is standard, by induction on $n$ we can show that there exists a trace $\sphase_0, \ldots, \sphase_n$ of the phase automaton such that
	$\sigma_i, v \models \statelabel{\sphase_i}$ and $(\sigma_i, \sigma_{i+1}), v \models \edgelabel{(\sphase_i, \sphase_{i+1})}$.
	This implies that $\ov{\sigma}, v \models \phaseauto$, as required.
\qed
\end{proof}

Equipped with the notion of simulation, we now return to the proof of \Cref{lemma:phase-ind-to-phase-inv}, and show that the definition of an inductive phase invariant ensures that
$\statelabel{}$ induces a simulation relation, and therefore ensures that the automaton is a phase invariant.
\begin{proof}[Proof of \Cref{lemma:phase-ind-to-phase-inv}]
We show that if $\phaseauto$ is inductive w.r.t.\ $\TS$ then $\phaseauto$ simulates $\TS$ (and hence by \Cref{lem:sim-to-phase-inv} it is a phase invariant).	
Fix a domain $\Dom$ and a valuation $v$.
	Take the simulation relation to be
	\begin{equation}
		H = \{(\sigma, \sphase) \, | \, \sigma, v \models \statelabel{\sphase}\}.
	\end{equation}
	To see that this is a simulation relation as required in \Cref{def:simulation}, note that requirement \ref{it:simulation-labeling-agree} holds by construction.
	The condition of \Cref{eq:ind-auto-init} implies requirement \ref{it:simulation-init}.
	For requirement \ref{it:simulation-next}, let $(\sigma, \sphase) \in H$ and $(\sigma, \sigma') \in \TR$.
	We have that $\sigma, v \models \statelabel{\sphase}$, so, from edge covering (\Cref{eq:ind-edge-covering}), there exists $\tphase \in \autostate$ such that $(\sigma, \sigma'), v \models \edgelabel{(\sphase,\tphase)}$.
	From the condition of \Cref{eq:ind-auto-consec} necessarily $\sigma', v \models \statelabel{\tphase}'$, and thus $(\sigma', \tphase) \in H$, as required.
\end{proof}

}

\OMIT{
\begin{corollary}
\label{cor:ind-safe-implies-safe}
If $\phaseauto$ is inductive w.r.t.\ $\TS$ and safe w.r.t\ a safety property $\forall \qoset. \ \Safety$ then $\forall \qoset. \, \Safety$ is an invariant of $\TS$.
\end{corollary}
}

\begin{remark} \label{rem:weaker-ind}
The careful reader may notice that the inductiveness requirement is stronger than needed to ensure that the characterizations form a phase invariant. It could be weakened to require for every $\sphase \in \autostate$:
$
\forall \qoset. \ \statelabel{\sphase} \wedge \TR \implies \bigvee_{(\sphase,\tphase) \in \autotrans}{\edgelabel{(\sphase,\tphase)} \wedge \statelabel{\tphase}'}.
$
However, as we explain in \Cref{sec:inference}, our notion of inductiveness is crucial for \emph{inferring} inductive phase automata, which is the goal of this paper.
Furthermore, for deterministic phase automata, the two requirements coincide.
\end{remark}

\para{Inductive invariants vs. inductive phase invariants}
Inductive invariants and inductive phase invariants are closely related:
\iflong
an inductive phase invariant  induces a ``standard'' inductive invariant, and vice versa: \fi
\begin{lemma}
\label{lemma:phase-to-ind}
If $\phaseauto$ is inductive w.r.t.\ $\TS$ then
$
	\forall \qoset. \ \bigvee_{\sphase \in \autostate}{\statelabel{\sphase}}
$
is an inductive invariant for $\TS$. \iflong
Conversely, if
\else
If
\fi
$\Inv$ is an inductive invariant for $\TS$, then the phase automaton $\phaseauto_{\Inv} = (\{\sphase\}, \{\sphase\}, \emptyset, \edgelabel{}, \statelabel{})$, where $\edgelabel{(\sphase,\sphase)} = \TR$ and $\statelabel{\sphase} = \Inv$ is an inductive phase automaton w.r.t. $\TS$.
\end{lemma}
In this sense, phase inductive invariants are as expressive as inductive invariants.
However, as we show in this paper, their structure can be used by a user as an intuitive way to guide an automatic invariant inference algorithm.

\iflong
\begin{remark}
It is straightforward to add more flexibility to a phase automaton by allowing a set of initial states, $\autoinitset$.
In this case, for the automaton to over-approximate all the reachable states of $\TS$, it suffices that every initial state corresponds to some initial phase of $\TS$, possibly depending on the valuation of $\qoset$.
Therefore, the initiation constraint in the definition of an inductive phase invariant for $\TS$ may be relaxed into:
\[
\Init \implies \left(\forall \qoset. \ \bigvee_{\sphase_0 \in \autoinitset}\statelabel{\sphase_0}\right).
\]
\end{remark}
\fi

\iflong
\subsection{Safe Inductive Phase Invariants}
\else
\para{Safe Inductive Phase Invariants}
\fi
Next we show that an inductive phase invariant can be used to establish safety.

\begin{definition}[Safe Phase Automaton]
\label{def:safe-automaton}
Let $\phaseauto$ be a phase automaton over $\voc$ with quantifiers $\qoset$
\iflong, and let $\forall \qoset.\ \Safety$ be a safety property\fi. Then $\phaseauto$ is \emph{safe} w.r.t.\ $\forall \qoset. \ \Safety$ if
\iflong
\begin{equation*}
\label{eq:safe-automaton}
	\forall \qoset. \ \left(\statelabel{\sphase} \implies \Safety\right)
\end{equation*}
\else
\, $\forall \qoset. \ \left(\statelabel{\sphase} \implies \Safety\right)$
\fi
holds for every $\sphase \in \autostate$.
\end{definition}

\OMIT{
\iflong
The following lemma shows that if $\phaseauto$ is a phase invariant of $\TS$ then the safety of $\phaseauto$ implies safety of $\TS$.
\begin{lemma} \label{lem:phase-inv-to-safety}
If $\phaseauto$ is a phase invariant for $\TS$ and safe w.r.t.\ $\forall \qoset.\ \Safety$, then $\forall \qoset. \, \Safety$ is an invariant of $\TS$.
\end{lemma}
\begin{proof}
	Let $\sigma_1,\ldots,\sigma_n$ be a finite trace of $\TS$.
	Let $v$ be a valuation for $\qoset$.
	$\ov{\sigma} \models \phaseauto$, so there exists a trace of phases $\sphase_0,\ldots,\sphase_n$ such that $\sigma_i, v \models \statelabel{\sphase_i}$.
	Since $v \models \statelabel{\sphase_i} \to \Safety$, for every $i$ it holds that $\sigma_i, v \models \Safety$.
The claim follows.
\end{proof}

\begin{corollary}
\label{cor:ind-safe-implies-safe}
\iflong
If $\phaseauto$ is inductive w.r.t.\ $\TS$ and safe w.r.t.\ a \TODO{omit next 2 words if it helps} safety property $\forall \qoset. \ \Safety$ then $\forall \qoset. \, \Safety$ is an invariant of $\TS$.
\else
If $\phaseauto$ is inductive w.r.t.\ $\TS$ and safe w.r.t.\ $\forall \qoset. \ \Safety$ then $\forall \qoset. \, \Safety$ is an invariant of $\TS$.
\fi
\end{corollary}
}

\begin{lemma}
\label{cor:ind-safe-implies-safe}
If $\phaseauto$ is inductive w.r.t.\ $\TS$ and safe w.r.t.\ $\forall \qoset. \ \Safety$ then $\forall \qoset. \, \Safety$ is an invariant of $\TS$.
\end{lemma}
\fi

\commentout{
\subsection{Obtaining Inductive Phase Invariants by Strengthening} \label{sec:stregthening}

Given a phase invariant that is not inductive w.r.t.\ $TS$, we may establish that it is an invariant by \emph{strengthening} it into one that is inductive:

\begin{definition}[Strengthening]
\label{def:strengthening}
A phase automaton $\phaseauto_2$ is a \emph{strengthening} of a phase automaton $\phaseauto_1$ if $\phaseauto_1, \phaseauto_2$ have the same automaton structure, and the characterization of every phase is stronger in $\phaseauto_2$ than it is in $\phaseauto_1$, namely, for every $\sphase \in \autostate^1 = \autostate^2$,
\begin{equation}
	\forall \qoset. \ \left(\statelabel{\sphase}^2 \implies \statelabel{\sphase}^1\right)
\end{equation}
\end{definition}

\begin{lemma}
\label{lem:ind-strengthening-to-phase-inv}
If $\phaseauto$ has a strengthening $\phaseauto'$ that is inductive for $\TS$ then $\phaseauto$ is a phase invariant for $\TS$.
\end{lemma}
\iflong
\begin{proof}
Let $\ov{\sigma} = \sigma_0,\ldots,\sigma_n$ be a trace of $\TS$ starting from an initial state.
We show that $\ov{\sigma} \in \Lang{\phaseauto}$.
Let $v$ be a valuation of $\qoset$.
By \Cref{lemma:phase-ind-to-phase-inv}, $\phaseauto'$ is a phase invariant for $\TS$. Hence, there exists a trace $\sphase_0,\ldots,\sphase_n$ of phases that corresponds to $(\ov{\sigma},v)$ in $\phaseauto'$. We show that the same trace is also a witness for $\phaseauto$.
Since $\phaseauto$ has the same structure as $\phaseauto'$, we know that $\sphase_0 = \autoinit$ and $(\sigma_i, \sigma_{i+1}), v \models \edgelabel{(\sphase_i, \sphase_{i+1})}$ for every $0 \leq i < n$.
It remains to show that $\sigma_i, v \models \statelabel{\sphase_i}$ for every $0 \leq i \leq n$, but since $\phaseauto'$ is a strengthening of $\phaseauto$, this follows immediately from $\sigma_i, v \models \statelabel{\sphase_i}'$. \TODO{the prime here is confusing!}
\end{proof}
\fi

We note that $\phaseauto$ may be a phase invariant for $\TS$ but not have any strengthening that is inductive.
This may happen for two reasons.

First, as with standard inductive invariants, it is possible that the necessary strengthening of the phase characterizations is not expressible in the logic available to us.

Second, even if we assume an unrestricted language of phase characterizations, it is possible that the edge labeling is too permissive, thus adding transitions that are not necessary for the edge covering requirement.
Such ``redundant'' transitions may sometimes be harmless, but they may also violate preservation along some edge.
Namely, if no state that has such an outgoing transition can reach the corresponding phase $\sphase$ from previous phases, such violations
can be overcome by strengthening $\statelabel{\sphase}$ to exclude all states that have such an outgoing transition (assuming an unrestricted language of phase characterizations), thus disabling the problematic transition along the edge.
In these cases, the automaton can be strengthened into one that is inductive.
However, in other cases, such strengthening
would exclude states that \emph{can} reach phase $\sphase$ and as such would damage the inductiveness property along incoming edges of $\sphase$.
In such cases, the only way to disable problematic transitions along automaton edges is by strengthening the transition relation formulas (i.e., updating the automaton structure). Hence, no inductive strengthening exists.
The second reason for the absence of a strengthening has no counterpart in standard inductive invariants; it reflects the additional structure expressed by a phase automaton, which is enforced by our stronger definition of inductiveness (as opposed to the weaker definition mentioned in \Cref{rem:weaker-ind}).
Fortunately, this reason can be avoided by considering deterministic phase automata:

\begin{lemma} \label{lem:strengthening-exists}
Let $\phaseauto$ be a deterministic phase automaton, and assume an unrestricted language of phase characterizations. Then $\phaseauto$ is a phase invariant for $\TS$ if and only if it has a  strengthening $\phaseauto'$ that is inductive w.r.t.\ $\TS$.
\end{lemma}
\iflong
\begin{proof}[Proof (sketch)]
The implication from right to left is clear. Consider the other direction.
For a deterministic phase automaton, $\Lang{\TS} \subseteq \Lang{\phaseauto}$ if and only if there exists a simulation relation between $\TS$ and $\phaseauto$.
Furthermore, in this case, the inductiveness requirement coincides with the requirement that the phase characterizations induce a simulation relation. Hence, in this case, by defining the characterization of $\sphase$ to include all the states simulated by it, we obtain an inductive strengthening of $\phaseauto$.
\end{proof}
\fi

We point out that restricting our attention to deterministic phase automata does not lose generality in the context of safety verification since every inductive phase invariant, which are the ones we seek, can be translated into a deterministic one:

\begin{lemma}
Let $\phaseauto =(\autostate, \autoinit, \qoset, \edgelabel{}, \statelabel{}) $ be an inductive phase invariant w.r.t.\ $\TS$. Define an arbitrary total order, $<$, on $\autostate$, and define $\phaseauto' = (\autostate, \autoinit, \qoset, \edgelabel{}', \statelabel{})$ where
\[
\edgelabel{(\sphase,\tphase)}' = \edgelabel{(\sphase,\tphase)} \wedge \bigwedge_{\tphase' < \tphase} \neg \edgelabel{(\sphase,\tphase')}
\]
Then $\phaseauto'$ is a deterministic inductive phase invariant w.r.t.\ $\TS$.
\end{lemma}
\iflong
\begin{proof}
The definition of $\edgelabel{}'$ ensures that it is deterministic, and inherits the edge covering property from $\edgelabel{}$.
Initiation is not affected by the edge labeling, and inductiveness cannot be damaged by strengthening $\edgelabel{(\sphase,\tphase)}'$.
\end{proof}
\fi

We note that when the language of phase characterizations is restricted and  an inductive strengthening does not exist for this reason, it may be possible to obtain an inductive strengthening in the given language by changing the automaton structure.

}  
\section{Inference of Inductive Phase Invariants}
\label{sec:inference}
\iflockserv
\else
\fi
In this section we turn to the \emph{inference} of safe inductive phase invariants over a given phase structure, which guides the search\iflong for invariant\fi.
\commentout{
Providing a full specification of an inductive phase invariant may still be difficult.
We therefore consider the problem of inferring a safe inductive phase invariant from a phase structure. }
\iflong
This section defines the problem, shows that it can be reduced to a set of Constrained Horn Clauses, discusses the aspects by which a phase structure guides inference, and considers witnesses for the case that no solution exists.
\fi
Formally, the problem we target is:
\begin{definition}[Inductive Phase Invariant Inference]  Given a transition system $\TS = (\Init, \TR)$, a phase structure $\phasestruct = (\autostate, \autoinit, \qoset, \edgelabel{})$  and a safety property $\forall \qoset. \ \Safety$, all over $\voc$, \emph{find} a safe inductive phase invariant $\phaseauto$ for $\TS$ over the phase structure $S$, if one exists.
\end{definition}

\iflockserv
\TODO{}
\else
\begin{example} \label{ex:inference}
Inference of an inductive phase invariant is provided with the phase structure in \Cref{fig:kv-automaton}, which embodies an intuitive understanding of the different phases the protocol undergoes (see \Cref{ex:phase-structure}). 
The algorithm automatically finds phase characterizations forming a safe inductive phase invariant over the user-provided structure.
We note that inference is valuable even after a phase structure is provided: in the running example, finding an inductive phase invariant is not easy; in particular, the characterizations in \Cref{fig:kv-invariants} relate different parts of the state and involve multiple quantifiers.
\end{example}
\fi

\subsection{Reduction to Constrained Horn Clauses}
\label{sec:chc-reduction}
We view each unknown phase characterization, $\statelabel{\sphase}$, which we aim to infer for every $\sphase \in \autostate$, as a predicate $I_\sphase$.
The definition of a safe inductive phase invariant induces a set of second-order Constrained Horn Clauses (CHC) over $I_\sphase$: \begin{align}
	\label{eq:chc-init}
&\text{\bf Initiation.} && \Init \implies \left(\forall \qoset. \ I_{\autoinit}\right)	&& \\
		&\text{\bf Inductiveness.} \text{ \ For every } (\sphase,\tphase) \in \autotrans:
         && \forall \qoset. \ \left(I_\sphase \land \edgelabel{(\sphase,\tphase)} \implies I'_\tphase\right) \label{eq:chc-ind}	 &&	\\	
		&\text {\bf Edge Covering.} \text{ \ For every } \sphase \in \autostate:
          && \forall \qoset. \ \bigg(I_\sphase \land \TR \implies \bigvee_{(\sphase,\tphase) \in \autotrans}{\edgelabel{(\sphase,\tphase)}}\bigg) \label{eq:chc-edge-cover} && \\
		&\text{\bf Safety.} \ \text{For every } \sphase \in \autostate:  && \forall \qoset. \ \left(I_\sphase \implies \Safety\right) \label{eq:chc-safe} &&
\end{align}
where $\qoset$ denotes the quantifiers of $\phaseauto$.
\iflong
The number of constraints is $1 + |\autotrans| + 2|\autostate|$.
\ifsketch
\TODO{}
If $\phaseauto$ is nothing more than a phase structure, constraint \ref{eq:chc-strengthen} can be omitted, and the number of constraints is
$1 + |\autotrans| + 2|\autostate|$.
\fi

\fi
All the constraints are \emph{linear}, namely at most one unknown predicate appears at the lefthand side of each implication.

\commentout{Using CHC terminology, \Cref{eq:chc-init} defines ``facts'' (since no unknown predicate appears in the body);
\Cref{eq:chc-ind} defines ``rules''; and
\Cref{eq:chc-edge-cover,eq:chc-safe,eq:chc-strengthen} define ``queries'' (since no unknown predicates appear in the righthand side of the implications).}
Constraint (\ref{eq:chc-safe}) captures the original safety requirement, whereas (\ref{eq:chc-edge-cover}) can be understood as additional safety properties that are specified by the phase automaton (since no unknown predicates appear in the righthand side of the implications).
\OMIT{Omitting: It requires that a system trace that reaches phase $\sphase$ with state
$\sigma$ has a matching outgoing edge of $\sphase$ for every outgoing transition of $\sigma$.}

A \emph{solution} $\mathbf{I}$ to the CHC system associates each predicate $I_\sphase$ with a formula $\psi_\sphase$ over $\vocabulary$ (with $\fv{\psi_\sphase} \subseteq \qoset$) such that when $\psi_\sphase$ is substituted for $I_\sphase$, all the constraints are satisfied (i.e., the corresponding first-order formulas are valid).
A solution to the system induces a safe inductive phase automaton through characterizing each phase $\sphase$ by the interpretation of $I_\sphase$, and vice versa. Formally:
\OMIT{
\iflong
\begin{lemma}
Let $\mathbf{I}_\sphase$ be a solution to the CHC system.
Then $\phaseauto' = (\autostate, \autotrans, \autoinit, \qoset, \edgelabel{}, \statelabel{}')$ with $\statelabel{\sphase}' = \mathbf{I}_\sphase$ is a safe inductive phase invariant wrt.\ $\TS$ and $\Safety$ and strengthens $\phaseauto$.
\end{lemma}
\begin{proof}
$\phaseauto'$ is safe due to constraint \ref{eq:chc-safe}. Constraint \ref{eq:chc-edge-cover} implies that $\phaseauto'$ covers $\TR$, constraint \ref{eq:chc-ind} implies that it is inductive, and with constraint \ref{eq:chc-init} this implies that $\phaseauto'$ is inductive wrt.\ $\TS$. Constraint \ref{eq:chc-strengthen} implies that $\phaseauto'$ is a strengthening of $\phaseauto$.
\end{proof}
The converse also holds:
\begin{lemma}
A safe inductive phase invariant which is a strengthening of $\phaseauto$ induces a solution to the CHC system. \end{lemma}
\else
}
\begin{lemma}\label{lem:inference}
Let $\phaseauto = (\autostate, \autotrans, \autoinit, \qoset, \edgelabel{}, \statelabel{})$ with $\statelabel{\sphase} = {\mathbf{I}}_\sphase$. Then $\phaseauto$ is a safe inductive phase invariant wrt.\ $\TS$ and $\forall \qoset.\ \Safety$
if and only if $\mathbf{I}$ is a solution to the CHC system.
\end{lemma}

Therefore, to infer a safe inductive phase invariant over a given phase structure, we need to solve the corresponding CHC system.
In \Cref{sec:updr-phase} we explain our approach for doing so for the class of universally quantified phase characterizations.
Note that the weaker definition of inductiveness discussed in \Cref{rem:weaker-ind} would prevent the reduction to CHC as it would result in clauses that are \emph{not} Horn clauses.

\para{Completeness of inductive phase invariants}
There are cases where a given phase structure induces a safe phase invariant $\A$, but not an inductive one, making the CHC system unsatisfiable. However, a strengthening into an inductive phase invariant can always be used to prove that $\phaseauto$ is an invariant if
\begin{inparaenum}[(i)]
	\item the language of invariants is unrestricted, and
	\item the phase structure is deterministic, namely, does not cover the same transition in two outgoing edges. Determinism of the automaton does not lose generality in the context of safety verification since every inductive phase automaton can be converted to a deterministic one; non-determinism is in fact unbeneficial as it mandates the same state to be characterized by multiple phases (see also \Cref{rem:weaker-ind}).
\end{inparaenum}
These topics are discussed in detail in \refappendix{sec:completeness}.

\begin{remark}
\label{rem:phase-decomposition}
Each phase is associated with a set of states that can reach it, where a state $\sigma$ can reach phase $\sphase$ if there is a sequence of program transitions that results in $\sigma$ and can lead to $\sphase$ according to the automaton's transitions. This makes a phase structure different from a simple syntactical disjunctive template for inference, in which such semantic meaning is unavailable.
\end{remark}

\iflong
\begin{remark} \TODO{remove? (space)}
When the safety property is of the form $\forall \qoset. \ \Grd \to \psi$ where $\Grd, \psi \in \form{\voc}{\qoset}$, we sometime seek an inductive phase invariant where $\Grd$ guards the entire phase structure of the automaton.
This may be represented by splitting the initial phase into two: one whose characterization includes $\Grd$, and another ``dummy'' initial phase whose characterization is $\neg \Grd$. The dummy initial phase has a single self loop labeled $\TR$, whereas the other maintains the actual phase structure. Moreover, $\Grd$ is added to the characterization of all other phases.
\end{remark}
\fi

\subsection{Phase Structures as a Means to Guide Inference}
\label{sec:inference-benefit}
The search space of invariants over a phase structure is in fact \emph{larger} than that of standard inductive invariants, because each phase can be associated with different characterizations.
Sometimes the disjunctive structure of the phases (\Cref{lemma:phase-to-ind}) uncovers a significantly simpler invariant than exists in the syntactical class of standard inductive invariants explored by the algorithm, but this is not always the case.\footnote{
As an illustration, \refappendix{sec:kv-ind} includes an inductive invariant for the running example which is comparable in complexity to the inductive phase invariant in \Cref{fig:kv-invariants}.
}
Nonetheless, the search for an invariant over the structure is \emph{guided}, through the following aspects: \\
(1) \emph{Phase decomposition.}
Inference of an inductive phase invariant aims to find characterizations that overapproximate the set of states reachable in each phase (\Cref{rem:phase-decomposition}). 
The distinction between phases is most beneficial when there is a considerable \emph{difference} between the sets associated with different phases and their characterizations. For instance, in the running example, all states without unreceived transfer messages are associated with $\texttt{O[$k$]}$, whereas all states in which such messages exist are associated with $\texttt{T[$k$]}$---a distinction captured by the characterizations in \cref{kv-char:disallow-recv-transfer,kv-char:unique-transfer-unreceived} in \Cref{fig:kv-invariants}.

Differences between phases would have two consequences. First, since each phase corresponds to fewer states than all reachable states, generalization---the key ingredient in inference procedures---is more focused. The second consequence stems from the fact that inductive characterizations of different phases are correlated. It is expected that a certain property is more readily learnable in one phase, while related facts in other phases are more complex. For instance, the characterization in \cref{kv-char:disallow-recv-transfer} in \Cref{fig:kv-invariants} is more straightforward than the one in \cref{kv-char:unique-transfer-unreceived}. Simpler facts in one phase can help characterize an adjacent phase when the algorithm analyzes how that property evolves along the edge. Thus utilizing the phase structure can improve the gradual construction of overapproximations of the sets of states reachable in each phase.
\\
(2) \emph{Disabled transitions.} A phase automaton explicitly states which transitions of the system are enabled in each phase, while the rest are disabled. Such impossible transitions induce additional safety properties to be established by the inferred phase characterizations.
\iflockserv
For example, the phase structure in \Cref{fig:lockservice-phase} forbids a $\texttt{recv\_grant}(\texttt{c}, \lockelem)$ action in all phases but $\texttt{G}(\lockelem)$.
\else
For example, the phase invariant in \Cref{fig:kv-invariants} forbids a {\small \texttt{recv\_transfer\_message($k$)}} in $\texttt{O[$k$]}$, a fact that can trigger the inference of the characterization in \cref{kv-char:disallow-recv-transfer}.
\fi
These additional safety properties direct the search for characterizations that are likely to be important for the proof.
\\
(3) \emph{Phase-awareness.} Finally, while a phase structure can be encoded in several ways (such as ghost code), a key aspect of our approach is that the phase decomposition and disabled transitions are \emph{explicitly} encoded in the CHC system in \Cref{sec:chc-reduction}, ensuring that they guide the otherwise heuristic search.
\\
In \Cref{sec:eval-dissection} we demonstrate the effects of aspects (1)--(3) on guidance.

\newcommand{\epnruns}{16}
\newcommand{\eptimeout}{1 hour}

\newcommand{\eparch}{a 3.4GHz AMD Ryzen Threadripper 1950X with 16 physical cores, running Linux 4.15.0, using Z3 version 4.7.1}

\section{Implementation and Evaluation}
\label{sec:evaluation}
In this section we apply invariant inference guided by phase structures to distributed protocols modeled in EPR, motivated by previous deductive approaches\iflong to safety of distributed protocols\fi~\cite{pldi/PadonMPSS16,DBLP:journals/pacmpl/PadonLSS17,DBLP:conf/pldi/TaubeLMPSSWW18}.
\iflong
\begin{figure}[t]
  \centering
\includegraphics[width=\textwidth]{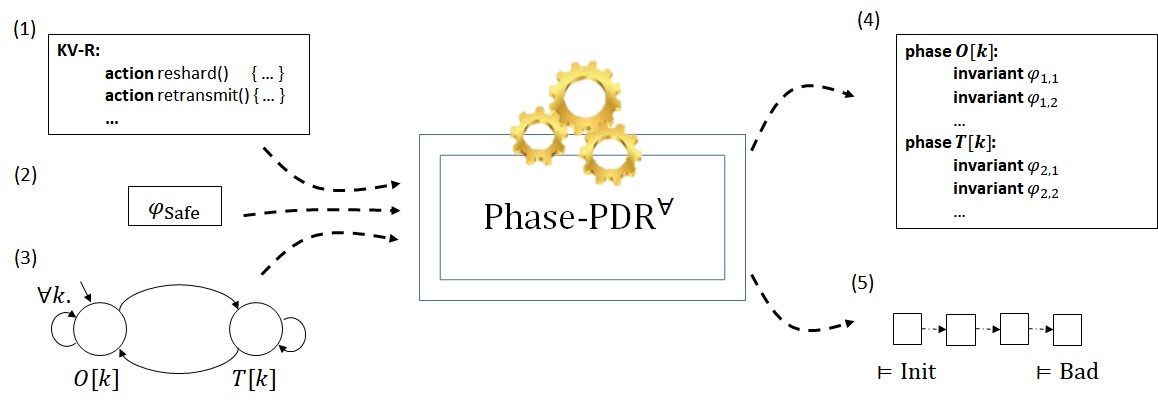}
  \caption{The work-flow of user-guided invariant inference through phase structures.
  The input to the tool is a (1) program, (2) safety property, and (3) phase structure.
  The output is either (4) inductive phase characterizations over the given phase structure, or (5) proof that universally quantified inductive phase characterizations do not exist, in the form of an abstract trace.
  }
  \label{fig:workflow}
\end{figure}
 The work-flow for our approach is illustrated in \Cref{fig:workflow}.
\fi
\subsection{Phase-$\UPDR$ for Inferring Universally Quantified Characterizations}
\label{sec:updr-phase}
We now describe our procedure for solving the CHCs system of \Cref{sec:chc-reduction}.
It either (i) returns universally quantified phase characterizations that induce a safe inductive phase invariant, (ii) returns an abstract counterexample trace demonstrating that this is not possible, or (iii) diverges.

\para{EPR}
Our procedure handles transition systems expressed using the extended {\bf E}ffectively {\bf PR}opositional fragment (EPR) of first order logic~\cite{epr,pldi/PadonMPSS16}, and infers universally quantified phase characterizations.
Satisfiability of (extended) EPR formulas is decidable, enjoys the finite-model property, and supported by 
\iflong
existing SMT solvers such as Z3~\cite{de2008z3} and first order logic provers such as iProver~\cite{DBLP:conf/cade/Korovin08}.
\else
solvers such as Z3~\cite{de2008z3} and iProver~\cite{DBLP:conf/cade/Korovin08}.
\fi

\para{Phase-$\UPDR$}
Our procedure is based on $\UPDR$~\cite{DBLP:journals/jacm/KarbyshevBIRS17}, a variant of PDR~\cite{ic3,pdr} that infers universally quantified inductive invariants.
PDR computes a sequence of \emph{frames} $\Frame_0, \ldots, \Frame_n$
such that $\Frame_i$ overapproximates the set of states reachable in $i$ steps.
In our case, each frame $\Frame_i$ is a mapping from a phase $\sphase$ to characterizations.
The details of the algorithm are standard for PDR; we describe the gist of the procedure in \refappendix{sec:algorithm-long}. We only stress the following:
Counterexamples to safety take into account the safety property as well as disabled transitions. Search for predecessors is performed by going backwards on automaton edges, blocking counterexamples from preceding phases to prove an obligation in the current phase. Generalization is performed w.r.t.\ all incoming edges. As in $\UPDR$, proof obligations are constructed via diagrams~\cite{chang1990model}; in our setting these include the interpretation for the view quantifiers (see 
\ifapp
\refappendix{sec:algorithm-long}
\else
\cite{extended-version}
\fi 
for details).

\para{Edge covering check in EPR}
\iflong
When the transition relation formula is in EPR and the phase characterizations are universally quantified, the checks induced by \Cref{eq:chc-init,eq:chc-ind,eq:chc-safe} translate to checking (un)satisfiability of EPR formulas, potentially causing divergence of the solver.
\Cref{eq:chc-edge-cover} is trickier as checking implication between two EPR transition relations falls outside of EPR.
To use a decidable edge covering check, we exploit the typical structure of transition relations in our setting, which is a disjunction between the transition relation of exported actions (the different {\small \texttt{action}}s in \Cref{fig:kv-ivy}).
In the phase automaton we label an edge $(\sphase,\tphase)$ by a set of exported actions, each action $a$ with a guard $g^{(\sphase,\tphase)}_a$ which is an alternation-free formula (a Boolean combination of universal and existential closed formulas). The edge covering check (\Cref{eq:chc-edge-cover}) can now be written
\begin{equation}
	\forall \qoset. \ \bigg(I_\sphase \land \TR_a \implies \bigvee_{(\sphase,\tphase) \in \autotrans}{g^{(\sphase,\tphase)}_a}\bigg).
\end{equation}
The righthand side of the implication is alternation-free and thus the check falls into the decidable EPR class.
\else
In our setting, \Cref{eq:chc-init,eq:chc-ind,eq:chc-safe} fall in EPR, but not \Cref{eq:chc-edge-cover}.
Thus, we restrict edge labeling so that each edge is labeled with a $\TR$ of an {\small \texttt{action}}, 
together with an alternation-free precondition. It then suffices to check implications between the preconditions and the entire $\TR$ (see the extended version~\cite{extendedVersion}).
\fi
Such edge labeling is sufficiently expressive for \iflong the phase structures of\fi all our examples.
Alternatively, sound but incomplete bounded quantifier instantiation~\cite{DBLP:conf/tacas/FeldmanPISS17} could be used, potentially allowing more complex decompositions of $\TR$. 

\para{Absence of Inductive Phase Characterizations}
\label{sec:infer-absence}
What happens when the user gets the automaton wrong?
One case is when there does not exist an inductive phase invariant with universal phase characterizations over the given structure. When this occurs, our tool can return an \emph{abstract counterexample trace}---a sequence of program transitions and transitions of the automaton (inspired by~\cite{DBLP:journals/jacm/KarbyshevBIRS17,DBLP:conf/popl/PadonISKS16})---which constitutes a proof
of that fact
(see \refappendix{sec:abstract-traces}). The counterexample trace can assist the user in debugging the automaton or the program and modifying them.
For instance, missing edges occurred frequently when we wrote the automata of \Cref{sec:evaluation}, and we used the generated counterexample traces to correct them.

Another type of failure is when an inductive phase invariant exists but the automaton does not direct the search well towards it. In this case the user may decide to terminate the analysis and articulate a different intuition via a different phase structure. In standard inference procedures, the only way to affect the search is by modifying the transition system; instead, phase structures equip the user with an ability to guide the search.

\subsection{Evaluation}
We evaluate our approach for user-guided invariant inference 
\iflong
based on phase structures 
\fi
by comparing Phase-$\UPDR$ to standard $\UPDR$\iflong, its inductive invariant inference counterpart\fi.
We implemented $\UPDR$ and Phase-$\UPDR$ in \Tool~\cite{mypyvy}, a new system for invariant inference inspired by Ivy~\cite{DBLP:conf/sas/McMillanP18}, over Z3~\cite{de2008z3}.
We study\iflong the following effects of guidance by phase structures\fi:
\begin{enumerate}
  \item Can Phase-$\UPDR$ \empheval{converge} to a proof when $\UPDR$ does not (in reasonable time)? 

  \item Is Phase-$\UPDR$ \empheval{faster} than $\UPDR$? \ifsketch
  Do additional phase characterizations in a phase sketch yield further speedup?
  \fi

\OMIT{
\item How do Phase-$\UPDR$ and $\UPDR$ \empheval{scale} with the number of phases in the protocol? (\Cref{sec:eval-scalability})
}

  \item Which aspects of Phase-$\UPDR$ contribute to its performance benefits? \end{enumerate}
\ifsketch
In each protocol, we compare Phase-$\UPDR$ with a phase structure and
Phase-$\UPDR$ with a phase sketch which includes partial characterizations,
to the baseline $\UPDR$ inferring a standard inductive invariant.
\fi

\para{Protocols}
We applied $\UPDR$ and Phase-$\UPDR$ to the most challenging examples admitting universally-quantified invariants, which previous works verified using deductive techniques.
The protocols we analyzed are listed below and in \Cref{tab:examples-evaluation-table}. The full models appear in~\cite{mypyvy-examples}. 
The KV-R protocol analyzed is taken from one of the two realistic systems studied by the IronFleet paper~\cite{IronFleet} using deductive verification.

\para{Phase structures}
The phase structures we used appear in~\cite{mypyvy-examples}.
In all our examples, it was straightforward to translate the existing high-level intuition of important and relevant distinctions between phases in the protocol into the phase structures we report. For example, it took us less than an hour to finalize an automaton for KV-R. 
We emphasize that phase structures do not include phase characterizations; the user need not supply them, nor has to understand the inference procedure. Our exposition of the phase structures below refers to an intuitive meaning of each phase, but this is not part of the phase structure provided to the tool.

\OMIT{
\paragraph{Specifying phase structures} In our experience, producing phase structures exercises existing intuition about the protocol's correctness.
Abstract counterexample traces aided us in articulating the phase structures, in most cases indicating an omitted edge. \TODO{is this a good place for this?}
}

\iflong
\newcolumntype{S}{>{\scriptsize}c}

\begin{table}[t]
\centering
\begin{threeparttable}\centering
\footnotesize
\begin{tabular}{|l|c|c||c|c||c|c||cccc|cccc|}
\hline
\multirow{2}{*}{Program} & \multirow{2}{*}{$\UPDR$} & \multirow{2}{*}{Phase-$\UPDR$} & \multirow{2}{*}{\#p} & \multirow{2}{*}{\#v} & \multirow{2}{*}{\#r} & \multirow{2}{*}{|a|}
& \multicolumn{4}{c|}{Inductive} & \multicolumn{4}{c|}{Phase-Inductive}
\\
& & & & & & & |f| & \#c & \#q & \#l & |f| & \#c & \#q & \#l
\\
\hline
\begin{tabular}{l}
Lock service \\
(single lock)\end{tabular}        & 2.21 (00.03)        & 0.67 (0.01)     & 4   & 1 &  5    & 1

& 11 & 9 & 15 & 21
& 6 & 3--4 & 3--4 & 3--6
\\
\begin{tabular}{l}
Lock service\\(multiple locks)\end{tabular}     & 2.73 (00.02)        & 1.06 (0.01)     & 4   & 1 &  5    & 2

& 11 & 9 & 24 & 21
& 6 & 4 & 3--4 & 4--6
\\
Consensus                  & 60.54 (2.95)        & 1355 (570)$^*$  & 3   & 1 &  7 & 2

& 9 & 6 & 15 & 15 
& 12 & 5--6 & 10--14 & 9--15
\\
KV (basic)                & 1.79 (0.02)         & 1.59 (0.02)     & 2   & 1 &  3 & 3

& 5 & 7 & 27 & 19 
& 5 & 4 & 9--10 & 8--9
\\
Ring leader              & 152.44 (39.41)      & 2.53 (0.04)     & 2   & 2 &  4 & 3

& 6--7 & 6 & 11 & 16 
& 5 & 1--2 & 0--1 & 1--4
\\
\hline
KV-R   & 2070 (370)$^*$      & 372.5 (35.9)    & 2   & 1 &  7 & 5

& 12--15 & 24 & 156 & 106 
& 11--13 & 5--11 & 15--67 & 12--52
\\
Cache coherence              & > 1 hour                 & 90.1 (0.82)     & 10  & 1 &  11 & 2

& n/a & n/a & n/a & n/a 
& 13 & 10--15 & 12--27 & 14--39
\\
\hline
\end{tabular}
\end{threeparttable}
	\caption{
	\label{tab:examples-evaluation-table}
\footnotesize Running times in seconds of $\UPDR$ and Phase-$\UPDR$, presented as
          the mean and standard deviation (in parentheses) over 16 different Z3 random seeds.
          ``${}^*$'' indicates that some runs did not converge after 1 hour and were not included in the summary statistics.
          ``> 1 hour'' means that no runs of the algorithm converged in 1 hour.
          \#p refers to the number of phases and \#v to the number of view quantifiers in the phase structure.
          \#r refers to the number of relations and |a| to the maximal arity.
          The remaining columns describe the inductive invariant/phase invariant obtained in inference. |f| is the maximal frame reached. \#c, \#q, \#l are the mean number of clauses, quantifiers (excluding view quantifiers) and literals per phase, ranging across the different phases.
}
\end{table}
 \else
\newcolumntype{S}{>{\scriptsize}c}

\begin{table}[t]
\centering
\begin{threeparttable}\centering
\footnotesize
\begin{tabular}{|l|c|c||c|c||c|c||ccc|ccc|}
\hline
\multirow{2}{*}{Program} & \multirow{2}{*}{$\UPDR$} & \multirow{2}{*}{Phase-$\UPDR$} & \multirow{2}{*}{\#p} & \multirow{2}{*}{\#v} & \multirow{2}{*}{\#r} & \multirow{2}{*}{|a|}
& \multicolumn{3}{c|}{Inductive} & \multicolumn{3}{c|}{Phase-Inductive}
\\
& & & & & & & |f| & \#c & \#q & |f| & \#c & \#q
\\
\hline
\begin{tabular}{l}
Lock service \\
(single lock)\end{tabular}        & 2.21 (00.03)        & 0.67 (0.01)     & 4   & 1 &  5    & 1

& 11 & 9 & 15
& 6 & 3--4 & 3--4
\\
\begin{tabular}{l}
Lock service\\(multiple locks)\end{tabular}     & 2.73 (00.02)        & 1.06 (0.01)     & 4   & 1 &  5    & 2

& 11 & 9 & 24
& 6 & 4 & 3--4
\\
Consensus                  & 60.54 (2.95)        & 1355 (570)$^*$  & 3   & 1 &  7 & 2

& 9 & 6 & 15
& 12 & 5--6 & 10--14
\\
KV (basic)                & 1.79 (0.02)         & 1.59 (0.02)     & 2   & 1 &  3 & 3

& 5 & 7 & 27
& 5 & 4 & 9--10
\\
Ring leader              & 152.44 (39.41)      & 2.53 (0.04)     & 2   & 2 &  4 & 3

& 6--7 & 6 & 11
& 5 & 1--2 & 0--1
\\
\hline
KV-R   & 2070 (370)$^*$      & 372.5 (35.9)    & 2   & 1 &  7 & 5

& 12--15 & 24 & 156 
& 11--13 & 5--11 & 15--67
\\
Cache coherence              & > 1 hour                 & 90.1 (0.82)     & 10  & 1 &  11 & 2

& n/a & n/a & n/a
& 13 & 10--15 & 12--27
\\
\hline
\end{tabular}
\end{threeparttable}
	\caption{
	\label{tab:examples-evaluation-table}
\footnotesize Running times in seconds of $\UPDR$ and Phase-$\UPDR$, presented as
          the mean and standard deviation (in parentheses) over 16 different Z3 random seeds.
          ``${}^*$'' indicates that some runs did not converge after 1 hour and were not included in the summary statistics.
          ``> 1 hour'' means that no runs of the algorithm converged in 1 hour.
          \#p refers to the number of phases and \#v to the number of view quantifiers in the phase structure.
          \#r refers to the number of relations and |a| to the maximal arity.
          The remaining columns describe the inductive invariant/phase invariant obtained in inference. |f| is the maximal frame reached. \#c, \#q are the mean number of clauses and quantifiers (excluding view quantifiers) per phase, ranging across the different phases. 
}
        \vspace{-0.5cm}
\end{table}
 \fi

\subsubsection{(1) Achieving Convergence Through Phases}
\label{sec:eval-convergence}
In this section we consider the effect of phases on inference for examples on which standard $\UPDR$ does not converge in \eptimeout{}.

\para{Examples}
\iflockserv
\emphexample{Sharded key-value store with retransmissions (KV-R)}, based on IronKV from IronFleet~\cite[\S 5.2.1]{IronFleet}. 
  This is a distributed hash table where each node owns some subset of the keys,
  and keys can be dynamically transferred among nodes to balance load.
  The safety property ensures that each key is globally associated with one value, even in the presence of key transfers.
  Messages might be dropped by the network, and the protocol uses retransmissions and sequence numbers to maintain availability and safety.
  This example is especially challenging for invariant inference in part due to the high arity of relations in its model (two relations have arity 5).
  The \empheval{phase structure} takes the view of a single key in the store, and has two phases: \emph{owned} and \emph{transferring}, intuitively meaning that the key is owned by a node or that an ownership transfer is in progress, respectively.
  The automaton transitions from \emph{owned} to \emph{transferring} upon the system's action in which a node cedes ownership of a key and sends a message (with a sequence number) to the new owner-to-be. The automaton transitions from \emph{transferring} to \emph{owned} when the target node receives the message (and it is not a duplicate).
  \ifsketch
  The \empheval{phase sketch} we consider states, in the form of partial characterizations, that
  in \emph{owned} the key is owned by at most one node and there are no pending messages concerning this key; in phase \emph{transferring}, the key is not owned by any node, and pending non-duplicate transfer messages are unique.
  \fi
\else
\emphexample{Sharded key-value store with retransmissions (KV-R)}: see \Cref{sec:running} and \Cref{ex:phase-structure}. This protocol has not been modeled in decidable logic before.
\fi

\emphexample{Cache coherence}. This example implements the classic MESI protocol for maintaining cache coherence in a shared-memory multiprocessor~\cite{hennessy-patterson-arch-6ed}, modeled in decidable logic for the first time.
  Cores perform reads and writes to memory, and caches snoop on each other's requests using a shared bus and maintain the invariant that there is at most one writer of a particular cache line.
  For simplicity, we consider only a single cache line, and yet the example is still challenging for $\UPDR$.
  Standard explanations of this protocol in the literature already use automata to describe this invariant, and we directly exploit this structure in our phase automaton.
  \empheval{Phase structure}:
  There are 10 phases in total, grouped into three parts corresponding to the modified, exclusive, and shared states in the classical description.
  Within each group, there are additional phases for when a request is being processed by the bus.
  For example, in the shared group, there are phases for handling reads by cores without a copy of the cache line, writes by such cores, and also writes by cores that \emph{do} have a copy.
  Overall, the phase structure is directly derived from textbook descriptions, taking into account that use of the shared bus is not atomic.

\OMIT{
  \item \emphexample{Single-decree Paxos}, based on the model from~\cite{DBLP:journals/pacmpl/PadonLSS17}; see~\cite{DBLP:journals/pacmpl/PadonLSS17} for an exposition to the protocol in EPR. The ``choosable'' relations is provided as a derived relation to reduce the invariant to a universal form---without this, Phase-$\UPDR$ returns an abstract counterexample trace. \TODO{also true for $\UPDR$}
  The \empeval{phase structure} takes the view of two values $v_1,v_2$ as they compete to become proposed in the highest possible round. The automaton transitions between phases in which the highest round (of rounds proposing either $v_1$ or $v_2$) has proposed $v_1$, and a symmetric one in which the highest round has proposed $v_2$.
  The \empheval{phase structure} we consider states exactly these facts about the proposal of the highest round.
  Phase-$\UPDR$ automatically infers additional required characterizations, stating that when the proposal from the highest round belongs to $v_1$ then (1) $v_2$ is not choosable in lower rounds, and that (2) $v_1$ itself it not choosable in rounds even lower than the highest round proposing $v_2$.
        We note that the automaton's structure conveys a slightly different intuition Paxos's safety proof than the standard standpoint of inductive invariant based-proof in previous works.
      \item
}

\para{Results and discussion}
Measurements for these examples appear in \Cref{tab:examples-evaluation-table}.
Standard $\UPDR$ fails to converge in less than an hour on 13 out of 16 seeds for KV-R and all 16 seeds for the cache\iflong coherence protocol\fi.
In contrast, Phase-$\UPDR$ converges to a proof in a few minutes in all cases.
\ifsketch
(Discussion of the improvement gained from partial characterizations in these examples is deferred to the next section.)
\fi
These results demonstrate that phase structures can effectively guide the search and obtain an invariant quickly where standard inductive invariant inference does not.

\subsubsection{(2) Enhancing Performance Through Phases}
\label{sec:eval-performance}
In this section we consider the use of phase structures to improve the speed of convergence to a proof.

\para{Examples}
\iflockserv
\emphexample{Distributed lock service}, which is described in detail in \Cref{sec:running} and \Cref{ex:phase-structure}.
  \ifsketch
  The phase sketch includes in each phase its main invariant: the first invariant in each phase in \Cref{fig:lockserv-invariants}.
  \fi
  We also consider a simpler variant with only a single lock,
  which reduces the arity of all relations and removes the need for a (nontrivial) automaton view. Its \empheval{phase structure} is the same, only for a single lock.
\else
\emphexample{Distributed lock service}, adapted from~\cite{verdi-pldi},
allows clients to acquire and release locks by sending requests to
a central server, which guarantees that only one client holds each lock at a
time. 
\empheval{Phase structure}: for each lock, the phases follow the 4 steps by which a client completes a cycle of acquire and release.
We also consider a simpler variant with only a single lock,
  reducing the arity of all relations and removing the need for an automaton view. Its \empheval{phase structure} is the same, only for a single lock.
\fi

\emphexample{Simple quorum-based consensus}, based on the example in~\cite{DBLP:conf/pldi/TaubeLMPSSWW18}.
  In this protocol, nodes propose themselves and then receive votes from other nodes.
  When a quorum of votes for a node is obtained, it becomes the leader and decides on a value.
  Safety requires that decided values are unique.
  The \empheval{phase structure} distinguishes between the phases before any node is elected leader, once a node is elected, and when values are decided.
  Note that the automaton structure is unquantified.\ifsketch
  The \empheval{phase sketch} states that first there are no leaders and no decisions, then unique leaders and no decisions, and then unique leaders and unique decisions. The rest of the inductive phase invariant relates quorums, leaders and decisions, and votes and vote messages.
  \fi

\emphexample{Leader election in a ring}~\cite{chang1979improved,pldi/PadonMPSS16},
  in which nodes are organized in a directional ring topology with unique IDs,
  and the safety property is that an elected leader is a node with the highest ID.
  \empheval{Phase structure}: for a view of two nodes $n_1,n_2$, 
  in the first phase, messages with the ID of $n_1$ are yet to advance in the ring past $n_2$, while
  in the second phase, a message advertising $n_1$ has advanced past $n_2$.
  \ifsketch
  We did not consider a phase sketch in this example due to the small number of clauses in the inductive phase invariant. \yotam{does it sound bad?}
Phase-$\UPDR$ automatically infers the characterizations that
  in the first phase $n_1$ cannot be a leader and all messages with the ID of $n_1$ are between $n_1$ and $n_2$ in the ring,
  and that in the second phase necessarily the ID of $n_1$ is greater than that of $n_2$ (because otherwise $n_2$ would have dropped the message).
  This latter characterization includes another quantifier on nodes, which constrains interference (see \Cref{sec:related}).
  \else
  The inferred characterizations include another quantifier on nodes, constraining interference (see \Cref{sec:related}).
  \fi

\emphexample{Sharded key-value store (KV)} is a simplified version of KV-R above, without message drops and the retransmission mechanism.
The \empheval{phase structure} is exactly as in \mbox{KV-R}, omitting transitions related to sequence numbers and acknowledgment.
This protocol has not been modeled in decidable logic before.

\para{Results and discussion}
We compare the performance of standard $\UPDR$ and Phase-$\UPDR$ on the above examples,
with results shown in~\Cref{tab:examples-evaluation-table}.
For each example, we ran the two algorithms on 16 different Z3 random seeds.
Measurements were performed on \eparch{}.
By disabling hyperthreading and frequency scaling and pinning tasks to dedicated cores,
variability across runs of a single 
seed was negligible.

In all but one example, Phase-$\UPDR$ improves performance, sometimes drastically; for example, performance for leader election in a ring is improved by a factor of 60.
Phase-$\UPDR$ also improves the \emph{robustness} of inference~\cite{DBLP:conf/cav/0001LMN14} on this example, as the standard deviation falls from 39 in $\UPDR$ to 0.04 in Phase-$\UPDR$.

The only example in which a phase structure actually diminishes inference effectiveness is simple consensus.
We attribute this to an automaton structure that does not capture the essence of the correctness argument very well, overlooking votes and quorums.
This demonstrates that a phase structure might guide the search towards counterproductive directions if the user guidance is ``misleading''.
This suggests that better resiliency of interactive inference framework could be achieved by combining phase-based inference with standard inductive invariant-based reasoning. We are not aware of a single ``good'' automaton for this example.
The correctness argument of this example is better captured by the conjunction of two automata (one for votes and one for accumulating a quorum) with different views, but the problem of inferring phase invariants for mutually-dependent automata is a subject for future work.
\ifsketch
We now turn to the effect of adding partial characterizations in a phase sketch on inference.
In almost all examples, additional phase characterizations provided by the user improve or perform as well as when just a phase structure is provided.
The phase sketch improve over a phase structure especially in the more challenging examples: KV-R and cache coherence, in which the additional characterizations reduce mean time from \textasciitilde 6 minutes to 3.5 and from 1.5 minutes to 44 seconds (respectively).
Broadly, the speedup obtained from adding additional phase characterizations to the sketch is weaker than the benefit of a phase structure over standard inductive invariant inference.
This suggests that in most cases the characterizations the user provides are in fact easy for the automatic algorithm to infer by itself from the phase structure, and thus do not offer much additional benefit. Speedup from partial characterizations (as apparent in KV-R and cache coherence) can occur when the user provides characterizations which are not readily apparent to the algorithm.
\fi

\OMIT{
  \input{figs/microbenchmark}
\para{Phase structures and scalability} In the following benchmark we study the effect of different phase structures and scalability.
\emphexample{$n$-phase commit} is an artificial protocol, based on 2-phase commit, achieving agreement between all nodes in $n$ rounds.
In each round, the coordinator sends a message to all participants and waits for unanimous acknowledgment
before proceeding to the next round.
The safety property is that if the final round has completed, then all nodes acknowledged in
the first round.
This property requires to reason about the entire protocol to associate all different rounds, so the complexity of the proof is in some sense proportional to $n$.
We compared two \empheval{phase structures}: (long) an automaton structure with one phase per round,
and (short) one with two phases, the first corresponding to round one and the second to the rest.
Both structures achieve a performance improvement as seen in \Cref{fig:microbenchmark}, with a preference to the long automaton.
The additional structure of phase invariants uncovers a disjunctive invariant that is not easily approximated by a conjunction of lemmas by $\UPDR$.
Somewhat surprisingly, the long automaton, which we find more intuitive, improves inference even further in spite of the overheard incurred by an increasing number of phases that must be handled.
}
\OMIT{
 Overall, the results suggest that harnessing the phase structure can improve the performance of invariant inference.
\ifsketch
, and that additional phase characterizations in the phase sketch potentially offer some additional improvement in certain cases.
\fi
 
}

\OMIT{
  almost all examples improved by phase structure

  comparing sketch versus structure: sometimes the sketch doesn't provide what algorithm needs,
  and adds some overhead (because it's an extra safety obligation)

  when we're worse, it's because the algorithm can already infer what the user gave.

  in other words, the difference between structure and sketch is explained by saying that
  the user provided something already inferrable from the phase structure
  (but if we're better than updr, then not by updr).

  this also suggests that it would be fruitful to spend additional engineering effort
  to reduce the overhead of additional safety checks.

  about toy consensus: not many disabled transitions,
  the phase characterizations aren't that different among phases

  calls for some combination of phase-based and inductive invariant based search
}

\subsubsection{(3) Anatomy of the Benefit of Phases}
\label{sec:eval-dissection}
We now demonstrate that each of the beneficial aspects of phases discussed in \Cref{sec:inference-benefit} is important for the benefits reported above. 

\para{Phase decomposition}
Is there a benefit from a phase structure even without disabled transitions? An example to a positive answer to this question is leader election in a ring, which demonstrates a huge performance benefit even without disabled transitions.

\para{Disabled transitions}
Is there a substantial gain from exploiting disabled transitions?
We compare Phase-$\UPDR$ on the structure with disabled transitions and a structure obtained by (artificially) adding self loops labeled with the originally impossible transitions, on the example of lock service with multiple locks (\Cref{sec:eval-performance}), seeing that it demonstrates a performance benefit using Phase-$\UPDR$ and showcases several disabled transitions in each phase.
The result is that without disabled transitions, the mean running time of Phase-$\UPDR$  on this example jumps from 2.73 seconds to 6.24 seconds.
This demonstrates the utility of the additional safety properties encompassed in disabled transitions.

\para{Phase-awareness}
Is it important to treat phases explicitly in the inference algorithm, as we do in Phase-$\UPDR$ (\Cref{sec:updr-phase})?
We compare our result on convergence of KV-R with an alternative in which standard $\UPDR$is applied to an encoding of the phase decomposition and disabled transition by \emph{ghost state}: each phase is modeled by a relation over possible view assignments, and the model is augmented with update code mimicking phase changes; the additional safety properties derived from disabled transitions are provided; and the view and the appropriate modification of the safety property are introduced.
This translation expresses all information present in the phase structure, but does not explicitly guide the inference algorithm to use this information.
The result is that with this ghost-based modeling the phase-oblivious $\UPDR$ does not converge in 1 hour on KV-R in any of the 16 runs, whereas it converges when Phase-$\UPDR$ explicitly directs the search using the phase structure.

\section{Related Work} \label{sec:related}
\iflong
\else
\vspace{-0.4cm}
\fi
\para{Phases in distributed protocols}
Distributed protocols are frequently described in informal descriptions as transitioning between different phases.
Recently, PSync~\cite{DBLP:conf/popl/DragoiHZ16} used the Heard-Of model~\cite{DBLP:journals/dc/Charron-BostS09}, which describes protocols as operating in rounds, as a basis for the implementation and verification of fault-tolerant distributed protocols.
Typestates~\cite[e.g.][]{DBLP:journals/tse/StromY86,DBLP:journals/scp/FieldGRY05} also bear some similarity to the temporal aspect of phases.
State machine refinement~\cite{DBLP:journals/tcs/AbadiL91,Garland:2000:UIA:336431.336455} is used extensively in the design and verification of distributed systems (see e.g.~\cite{Newcombe:2015:AWS:2749359.2699417,IronFleet}).
The automaton structure of a phase invariant is also a form of state machine; our focus is on inference of characterizations establishing this.

\para{Interaction in verification}
Interactive proof assistants such as Coq~\cite{DBLP:series/txtcs/BertotC04} and Isabelle/HOL~\cite{DBLP:books/sp/NipkowPW02} interact with users to aid them as they attempt to prove candidate inductive invariants. This differs from interaction through phase structures and counterexample traces.
Ivy uses interaction for invariant inference by interactive generalization from counterexamples~\cite{pldi/PadonMPSS16}. This approach is less automatic as it requires interaction for every clause of the inductive invariant.
In terminology from synthesis~\cite{DBLP:conf/synasc/Gulwani12}, the use of counterexamples is \emph{synthesizer-driven} interaction with the tool, while interaction via phase structures is mainly \emph{user-driven}. Abstract counterexample traces returned by the tool augment this kind of interaction.
As~\cite{DBLP:journals/acta/JhaS17} has shown, interactive invariant inference, when considered as a synthesis problem (see also~\cite{DBLP:conf/cav/0001LMN14,DBLP:journals/fmsd/SharmaA16}) is related to inductive learning.

\para{Template-based invariant inference} Many works employ syntactical templates for invariants, used to constrain the search~\cite[e.g.][]{DBLP:conf/cav/ColonSS03,DBLP:conf/sas/SankaranarayananSM04,DBLP:conf/pldi/SrivastavaG09,DBLP:journals/sttt/SrivastavaGF13,DBLP:series/natosec/AlurBDF0JKMMRSSSSTU15}. The different phases in a phase structure induce a disjunctive form, but crucially each disjunct also has a distinct semantic meaning, which inference overapproximates, as explained in \Cref{sec:inference-benefit}.

\para{Automata in safety verification}
Safety verification through an automaton-like refinement of the program's control has been studied in a number of works. We focus on related techniques for proof automation.
The \emph{Automizer} approach to the verification of sequential programs~\cite{DBLP:conf/sas/HeizmannHP09,DBLP:conf/cav/HeizmannHP13} is founded on the notion of a \emph{Floyd-Hoare automaton}, which is an unquantified inductive phase automaton; an extension to parallel programs~\cite{DBLP:conf/popl/FarzanKP15} uses thread identifiers closed under the symmetry rule,
which are related to view quantifiers.
Their focus is on the automatic, incremental construction of such automata as a union of simpler automata, where each automaton is obtained from generalizing the proof/infeasibility of a single trace.
In our approach the structure of the automaton is provided by the user as a means of conveying their intuition of the proof, while the annotations are computed automatically.
A notable difference is that in Automizer, the generation of characterizations in an automaton constructed from a single trace does not utilize the phase structure (beyond that of the trace), whereas in our approach the phase structure is central in generalization from states to characterizations.
\iflong

\fi
In \emph{trace partitioning}~\cite{DBLP:journals/toplas/RivalM07,DBLP:conf/esop/MauborgneR05}, abstract domains based on transition systems partitioning the program's control are introduced. The observation is that recording historical information forms a basis for case-splitting, as an alternative to fully-disjunctive abstractions. This differs from our motivation of distinguishing between different protocol phases.
The phase structure of the domain is determined by the analyser, and can also be dynamic. In our work the phase structure is provided by the user as guidance. We use a variant of $\UPDR$, rather than abstract interpretation~\cite{POPL:CC79}, to compute universally quantified phase characterizations.
\iflong

\fi
Techniques such as \emph{predicate abstraction}~\cite{DBLP:conf/cav/GrafS97,DBLP:conf/popl/FlanaganQ02} and \emph{existential abstraction}~\cite{DBLP:books/daglib/0007403}, as well as the safety part of \emph{predicate diagrams}~\cite{DBLP:conf/ifm/CansellMM00}, use finite languages for the set of possible characterizations and lack the notion of views, both essential for handling unbounded numbers of processes and resources.
\iflong

\fi
Finally, \emph{phase splitter predicates}~\cite{DBLP:conf/cav/SharmaDDA11} share our motivation of simplifying invariant inference by exposing the different phases the loop undergoes. Splitter predicates correspond to inductive phase characterizations~\cite[Theorem~1]{DBLP:conf/cav/SharmaDDA11}, and are automatically constructed according to program conditionals. In our approach, decomposition is performed by the user using potentially non-inductive conditions, and the inductive phase characterizations are computed by invariant inference. Successive loop splitting results in a sequence of phases, whereas our approach utilizes arbitrary automaton structures.
Borralleras et al.~\cite{DBLP:conf/tacas/BorrallerasBLOR17} also refine the control-flow graph throughout the analysis by splitting on conditions, which here are discovered as preconditions for termination (the motivation is to expose termination proof goals to be established): in a sense, the phase structure is grown from candidate characterizations implying termination. This differs from our approach in which the phase structure is used to guide the inference of characterizations.

\para{Quantified invariant inference}
We focus here on the works on quantifiers in automatic verification most closely related to our work.
In \emph{predicate abstraction}, quantifiers can be used internally as part of the definitions of predicates, and also externally through predicates with free variables~\cite{DBLP:conf/popl/FlanaganQ02,DBLP:journals/tocl/LahiriB07}.
Our work uses quantifiers both internally in phases characterizations and externally in view quantifiers.
The view is also related to the bounded number of quantifiers used in \emph{view abstraction}~\cite{DBLP:journals/sttt/AbdullaHH16,DBLP:conf/vmcai/AbdullaHH13}.
In this work we observe that it is useful to consider views of entities beyond processes or threads,
\iflockserv
such as a single lock in the running example.
\else
such as a single key in the store.
\fi
\iflong

\fi
Quantifiers are often used to their full extent in verification conditions, namely checking implication between two quantified formulas, but they are sometimes employed in weaker checks as part of thread-modular proofs~\cite{DBLP:journals/toplas/AbadiL95,Jones:1983:TST:69575.69577}.
This amounts to searching for invariants provable using specific instantiations of the quantifiers in the verification conditions~\cite{DBLP:conf/sigsoft/GurfinkelSM16,DBLP:conf/popl/HoenickeMP17}.
In our verification conditions, the view quantifiers are localized, in effect performing a single instantiation.
This is essential for exploiting the disjunctive structure under the quantifiers, allowing inference to consider a single automaton edge in each step, and reflecting an intuition of correctness.
When necessary to constrain interference, quantifiers in phase characterizations can be used to establish necessary facts about interfering views.
Finally, there exist algorithms other than $\UPDR$ for solving CHC by predicates with universal invariants~\cite[e.g.][]{DBLP:conf/atva/GurfinkelSV18,DBLP:conf/cav/DrewsA16}.

\OMIT{
Also Coq?
Houdini
ACL/2
Rahul Sharma
Dynamo
}

 \section{Conclusion}

Invariant inference techniques aiming to verify intricate distributed protocols must adjust to the diverse correctness arguments on which protocols are based. In this paper we have proposed to use phase structures as means of conveying users' intuition of the proof, to be used by an automatic inference tool as a basis for a full formal proof. We found that inference guided by a phase structure can infer proofs for distributed protocols that are beyond reach for state of the art inductive invariant inference methods, and can also improve the speed of convergence.
The phase decomposition induced by the automaton, the use of disabled transitions, and the explicit treatment of phases in inference, all combine to direct the search for the invariant.
We are encouraged 
by our experience of specifying phase structures for different protocols.
It would be interesting to integrate the interaction via phase structures with other verification methods and proof logics, as well as interaction schemes based on different, complementary, concepts.
Another important direction for future work is inference beyond universal invariants, required for example for the proof of Paxos~\cite{DBLP:journals/pacmpl/PadonLSS17}. 
\subsubsection*{Acknowledgements}
{\footnotesize
We thank
Kalev Alpernas, Javier Esparza, Neil Immerman, Oded Padon, Andreas Podelski, Tom Reps,
and
the anonymous referees
for insightful comments which improved this paper.
This publication is part of a project that has received funding from the European Research Council (ERC) under the European Union's Horizon 2020 research and innovation programme (grant agreement No [759102-SVIS]).
The research was partially supported by
Len Blavatnik and the Blavatnik Family foundation,
the Blavatnik Interdisciplinary Cyber Research Center, Tel Aviv University,
the Israel Science Foundation (ISF) under grant No. 1810/18,
the United States-Israel Binational Science Foundation (BSF) grant No. 2016260,
and the National Science Foundation under Grant No. 1749570. 
Any opinions, findings, and conclusions or recommendations expressed in this material are those of the authors and do not necessarily reflect the views of the National Science Foundation.
}

\bibliographystyle{splncs04}
\bibliography{refs}

\ifapp
\clearpage
\appendix
\section{Completeness of Inductive Phase Invariants}
\label{sec:completeness}
There are cases where a phase automaton $\phaseauto$ is a phase invariant for $\TS$, but this cannot be established via an inductive phase invariant since there is no strengthening of its phase characterizations that leads to an inductive phase invariant for $\TS$.
This may happen for two reasons.

First, as with standard inductive invariants, it is possible that the strengthening necessary  to ensure inductiveness is not expressible in the logic available to us.

Second, even if we assume an unrestricted language of phase characterizations, it is possible that the edge labeling is too permissive, thus adding transitions that are not necessary for the edge covering requirement.
Such ``redundant'' transitions may sometimes be harmless, but they may also violate preservation along some edge.
Namely, if no state that has such an outgoing transition can reach the corresponding phase $\sphase$ from previous phases, such violations
can be overcome by strengthening $\statelabel{\sphase}$ to exclude all states that have such an outgoing transition (assuming an unrestricted language of phase characterizations), thus disabling the problematic transition along the edge.
In these cases, an inductive automaton can be obtained.
However, in other cases, strengthening the phase characterizations in this way
would exclude states that \emph{can} reach phase $\sphase$ and as such would damage the inductiveness property along incoming edges of $\sphase$.
In such cases, the only way to disable problematic transitions along automaton edges is by strengthening the transition relation formulas (i.e., updating the automaton structure). Hence, no inductive automaton exists for the given phase structure.
The second reason has no counterpart in standard inductive invariants; it reflects the additional structure expressed by a phase automaton, which is enforced by our stronger definition of inductiveness (as opposed to the weaker definition mentioned in \Cref{rem:weaker-ind}).
Fortunately, this reason can be avoided by considering deterministic phase automata:

\begin{lemma} \label{lem:strengthening-exists}
Let $\phaseauto$ be a deterministic phase automaton, and assume an unrestricted language of phase characterizations. Then $\phaseauto$ is a phase invariant for $\TS$ if and only if it has a strengthening $\phaseauto'$ that is inductive w.r.t.\ $\TS$.
\end{lemma}
\iflong
\begin{proof}[Proof (sketch)]
The implication from right to left is clear. Consider the other direction.
For a deterministic phase automaton, $\Lang{\TS} \subseteq \Lang{\phaseauto}$ if and only if there exists a simulation relation between $\TS$ and $\phaseauto$.
Furthermore, in this case, the inductiveness requirement coincides with the requirement that the phase characterizations induce a simulation relation. Hence, in this case, by defining the characterization of $\sphase$ to include all the states simulated by it, we obtain an inductive strengthening of $\phaseauto$.
\end{proof}
\fi

Non-determinism is generally unbeneficial as it only mandates some states to be characterized by multiple phases in the inductive phase invariant (see also \Cref{rem:weaker-ind}).
We point out that restricting our attention to deterministic phase automata does not lose generality in the context of safety verification since every inductive phase invariant, which are the ones we seek, can be translated into a deterministic one, as the following lemma shows. Thus a structure admitting an inductive phase invariant can be converted to a deterministic one with the same property. \yotam{Sharon, can you review? Modified because one reviewer wrote ``you say that you can convert the automata into a deterministic one, but that comes with a price. The resulting automata of lemma 4.7 is not equivalent to the original one, while the text makes it sound that it is.''}
\begin{lemma} \label{lem:determinization}
Let $\phaseauto =(\autostate, \autoinit, \qoset, \edgelabel{}, \statelabel{}) $ be an inductive phase invariant w.r.t.\ $\TS$. Define an arbitrary total order, $<$, on $\autostate$, and define $\phaseauto' = (\autostate, \autoinit, \qoset, \edgelabel{}', \statelabel{})$ where
\[
\edgelabel{(\sphase,\tphase)}' = \edgelabel{(\sphase,\tphase)} \wedge \bigwedge_{\tphase' < \tphase} \neg \edgelabel{(\sphase,\tphase')}
\]
Then $\phaseauto'$ is a deterministic inductive phase invariant w.r.t.\ $\TS$.
\end{lemma}
\OMIT{
\begin{proof}
The definition of $\edgelabel{}'$ ensures that it is deterministic, and inherits the edge covering property from $\edgelabel{}$.
Initiation is not affected by the edge labeling, and inductiveness cannot be damaged by strengthening $\edgelabel{(\sphase,\tphase)}'$.
\end{proof}
}

We note that when the language of phase characterizations is restricted and an inductive phase automaton does not exist for this reason, it may be possible to overcome the limitations of the language and obtain one by changing the automaton structure.
 \section{Abstract Counterexample Traces}
\label{sec:abstract-traces}
Phase structures may not admit a safe inductive phase invariant.
In this section we discuss causes for this, and notions of \emph{concrete} and \emph{abstract} counterexample traces constituting a proof that inductive phase characterizations cannot be found over the given structure in the given language of candidate characterizations.

\commentout{
\subsubsection{Causes for Absence of Inductive Strengthening.}
\label{sec:absence-causes}
We start by listing the potential reasons for absence of an inductive strengthening.

First, if the transition system at hand is not safe w.r.t.\ the original safety property (enforced by \Cref{eq:chc-safe}), no safe inductive strengthening exists, for any phase automaton (by \Cref{cor:ind-safe-implies-safe}).

Second, even if the transition system is safe, some choices of phase automata are intrinsically incorrect, since they exclude true traces of the transition system.
This is the case when the ``additional'' safety properties specified by the phase automaton (\Cref{eq:chc-edge-cover,eq:chc-strengthen}) are erroneous (i.e., do not hold in some trace).
For example, consider a deterministic phase automaton $\phaseauto$ specified by the user. The phase automaton might allow a system trace to reach some phase $\sphase$ with valuation $v$ and state $\sigma$, but exclude it from the phase characterization ($\sigma,v \not\models \statelabel{\sphase})$ or exclude its outgoing transitions from all the outgoing edges of this phase ($(\sigma,\sigma') \in \TR$ but $(\sigma,\sigma'),v \not\models \edgelabel{(\sphase,\sphase')}$ for all $\sphase' \in \autotrans$). In such cases, the phase automaton does not overapproximate all the traces of $\TS$; it is in fact not a phase invariant, hence no inductive strengthening exists (by \Cref{lem:ind-strengthening-to-phase-inv},
an inductive strengthening is only possible when the phase automaton is a phase invariant), even though $\TS$ may be safe.

Finally, even if the given phase automaton $\phaseauto$ is a phase invariant, no inductive strengthening may exist for the two reasons discussed in \Cref{sec:stregthening}.
Specifically, assuming that $\phaseauto$ is deterministic, the remaining reason is that the language of phase characterizations is not expressive enough to capture the required strengthening.
}

\subsection{Concrete Counterexample Traces} \label{sec:conc-cex}

We first consider the case where no safe inductive phase invariant exists, regardless of the language of phase characterizations.
Such a case may be witnessed by a \emph{counterexample trace} that exhibits a violation of one of the safety properties induced by $\phasestruct$. 

\begin{definition}[Counterexample Trace]
A trace $\sigma_1,\ldots,\sigma_n$ of $\TS$ with a valuation $v$ for $\qoset$ is a \emph{counterexample trace} for $\phasestruct$, $\TS$ and $\forall \qoset. \ \Safety$ if
there exists a  trace $\sphase_0,\ldots,\sphase_n$ of $\phasestruct$ such that $({\sigma}_i,\sigma_{i+1}), v \models \edgelabel{(\sphase_i,\sphase_{i+1})}$ for every $1 \leq i < n$, but one of the following holds:
\begin{enumerate}
	\item \label{it:abs-cex-unsafe} A state in the trace is not safe:
		$\sigma_i, v \not\models \Safety$ for some $i$.
\item \label{it:abs-cex-untrans} A state in the  trace allows a transition that is not covered by any outgoing edge:
There exists $\sigma'$ s.t.\ $(\sigma_i,\sigma') \models \TR$ for some $i$, but no edge can cover this transition, i.e.\ $(\sigma_i,\sigma'), v \not\models \edgelabel{(\sphase_i,\sphase')}$ for all $\sphase' \in \autostate$.
\end{enumerate}
\end{definition}

\begin{lemma} \label{lem:cex-trace}
If a counterexample trace exists then there is no safe inductive phase invariant with structure $\phasestruct$, even if the language of phase characterizations is unrestricted. \end{lemma}

In case~\ref{it:abs-cex-unsafe}, the trace violates the original safety property (enforced by \Cref{eq:chc-safe}).
This means that the transition system at hand is not safe and no safe inductive phase invariant can be expected, for any phase structure (by \Cref{cor:ind-safe-implies-safe}).

In case~\ref{it:abs-cex-untrans}, the trace violates the ``additional'' safety property specified by the phase structure (\Cref{eq:chc-edge-cover}).
This does not indicate that the transition system is not safe, but rather that the choices of the phase structure are intrinsically incorrect: it includes a trace such that a corresponding trace of the transition system reaches a state that allows a transition that is not covered by any outgoing edge in the automaton. If the phase structure is nondeterministic, this may indicate that some of the edges are redundant and prevent obtaining phase inductive characterizations even if it is a phase invariant.
If the structure is deterministic (where every trace of $\TS$ corresponds to at most one trace of $\phaseauto$), this means that the automaton excludes true traces of the transition system and is therefore not a phase invariant, hence no corresponding inductive phase invariant exists. \TODO{automaton is not defined here}
In both cases, $\TS$ may be safe but the user has to modify the phase structure in order to be able to verify it.

\subsection{Abstract Counterexample Traces} In practice, phase characterizations are restricted since actual inference algorithms restrict their search space to a certain class of formulas.
This may be viewed as a form of abstraction, and may be one of the potential causes of absence of inductive phase invariant. In this case, we may obtain an \emph{abstract} counterexample trace that may indicate any of the violations discussed in \Cref{sec:conc-cex}, but may also reflect the limitations of the language of phase characterizations.

\para{Language of Phase Characterization}
We denote by $\Lset$ the class of formulas used to represent phase characterizations.
We denote by $\Lset_{\voc}(\qoset)$ the set of formulas in $\Lset$ over vocabulary $\voc$  with free variables from $\qoset$.
\iflong
We assume an implication relation over $\Lset$:
for every $\psi_1,\psi_2 \in \Lset_{\voc}(\qoset)$, a structure $\sigma$ over $\voc$ and a valuation $v$ we have the relation $\sigma, v \models \psi_1 \implies \sigma, v \models \psi_2$.
We use the implication relation to define a preorder on structures (accompanied by valuations) that captures which structure satisfies more formulas from $\Lset_{\voc}(\qoset)$:
\else
We define an preorder on structures paired by valuations that captures which structure satisfies more formulas from $\Lset_{\voc}(\qoset)$:
\fi

\begin{definition}
For finite structures over $\voc$ and valuation $v$ for $\qoset$, we define $(\sigma_1,v) \leqLof{\Lset_{\voc}(\qoset)} (\sigma_2,v)$ if $v$ is defined in both $\sigma_1$ and $\sigma_2$ and for all $\psi \in \Lset_{\voc}(\qoset)$, $\sigma_2, v \models \psi \Rightarrow \sigma_1, v \models \psi$.
\end{definition}
Intuitively, $(\sigma_1,v) \leqLof{\Lset_{\voc}(\qoset)} (\sigma_2,v)$ means that $(\sigma_2,v)$ is more abstract than $(\sigma_1,v)$: any formula in $\Lset_{\voc}(\qoset)$ that is satisfied by $(\sigma_2,v)$ is also satisfied by $(\sigma_1,v)$.
In particular, no formula that is satisfied by $(\sigma_2,v)$ can distinguish it from $(\sigma_1,v)$.
\iflong
Note that $\leqLof{\Lset_{\voc}(\qoset)}$ is defined for structures paired with the same valuation, which means that they interpret $\qoset$ in the same way.
\fi
We often omit $\voc$ and $\qoset$ from the notation and write $(\sigma_1,v) \leqLof{\Lset} (\sigma_2,v)$.

\begin{example}[Universal Characterizations] \label{ex:subtructure}
Consider $L = \Univ$, i.e., the class of universally quantified formulas.
In this case, $(\sigma_1,v) \leqU (\sigma_2,v)$ if 
$v$ is defined in both $\sigma_1,\sigma_2$ and 
$\sigma_1$ is a substructure of $\sigma_2$ (up to isomorphism).
(The structure $\sigma_1 = (\Dom_1,\Int_1)$ is a substructure of the structure $\sigma_2=(\Dom_2, \Int_2)$ if $\Dom_1 \subseteq \Dom_2$ and $\Int_2$ agrees with $\Int_1$ on $\Dom_1$.\iflong
That is, for every constant symbol $c \in \voc$, $\Int_2(c)=\Int_1(c)$, for every function symbol $f \in \voc$ with arity $k$, $\Int_2(f)(d_1,\ldots,d_k) = \Int_1(f)(d_1,\ldots,d_k)$ for every $d_1,\ldots,d_k \in \Dom_1$,
and for every relation symbol $r \in \voc$ with arity $k$, $\Int_2(r) \cap \Dom_1^k = \Int_1(r)$.
\fi )
\end{example}

\para{Abstract Traces}
We view the preorder $\leqLof{\Lset}$ as an abstraction relation, and use it to define a notion of an abstract trace, where each transition consists of an ``abstraction step'' followed by a concrete transition of the system. An abstraction step
transitions to a ``less abstract'' state (that cannot be distinguished by any formula satisfied by the source of the transition -- the more abstract state).

\begin{definition}[Abstract Trace]
Given a transition system $\TS = (\Init, \TR)$ over vocabulary $\voc$, an \emph{abstract trace} is a finite sequence of states ${\sigma}_1,\ldots,{\sigma_n}$ over $\voc$ with a valuation $v$ over $\qoset$ which is defined in all $\sigma_i$
such that for every $1 \leq i < n$, there exists $\tilde{\sigma_i}$ such that $(\tilde{\sigma}_i,v) \leqL (\sigma_i,v)$ and $(\tilde{\sigma}_i,\sigma_{i+1}) \models \TR$.
\end{definition}

Note that since $\leqL$ is reflexive, $\tilde{\sigma_i}$ may be equal to $\sigma_i$, in which case the abstract trace is concrete.

An \emph{abstract counterexample trace} is similar to a counterexample trace, except that it consists of an abstract trace. While the violation exhibited by an abstract counterexample trace may not be real, it indicates that no safe inductive phase invariant exists in $\Lset$.

\begin{lemma} \label{lem:abstract-cex-trace}
If an abstract counterexample trace exists then there is no safe inductive phase invariant with structure $\phasestruct$ and phase characterizations in $\Lset$.
\end{lemma}
\OMIT{
\begin{proof}
    For the sake of the proof, we explicitly split each transition in an abstract trace to an abstraction step followed by a transition.
Let $(\sigma_1,\tilde{\sigma}_1,\ldots,\sigma_n,\tilde{\sigma}_n),v)$ be such a trace and $\sphase_1,\ldots,\sphase_n$ be as in the definition.
    Let $\phaseauto'$ be a strengthening of $\phaseauto$ with phase characterization $\statelabel{}'$.
	Assume that $\phaseauto'$ is safe and show that $\phaseauto'$ is not inductive.
	Assume for the sake of contradiction that it is.
	We claim by induction on $i$ that $\sigma_i, v \models \statelabel{\sphase_i}'$.
	The base claim follows from initiation.
	For the induction step, assume that $\sigma_i, v \models \statelabel{\sphase_i}'$.
	Since $(\tilde{\sigma}_i, v) \leqL (\sigma_i, v)$ and $\statelabel{\sphase_i}' \in \Lset$, $\tilde{\sigma}_i, v \models \statelabel{\sphase_i}'$.
	Now $(\tilde{\sigma}_i, \sigma_{i+1}), v \models \edgelabel{(\sphase_i, \sphase_{i+1})}$, and from the assumption that $\phaseauto'$ is inductive necessarily $\sigma_{i+1}, v \models \statelabel{\sphase_{i+1}}'$, as required.

	Let us consider the cause of the abstract counterexample trace. There are three cases:
	In case \ref{it:abs-cex-unsafe}, $\sigma_i, v \not \models \Safety$ is a contradiction to the safety of $\phaseauto'$ (\Cref{eq:safe-automaton}).
	\ifsketch
	In case \ref{it:abs-cex-unchar}, $\sigma_i, v \not \models \statelabel{\sphase_i}$ implies that $\sigma_i, v \not \models \statelabel{\sphase_i}'$ (strengthening), which we have shown to be impossible.
	\fi
	In case \ref{it:abs-cex-untrans}, if $(\sigma_i, \sigma') \models \TR$, since we have shown $\sigma_i, v \models \statelabel{\sphase_i}'$ it must follow from the edge covering that
    $(\sigma_i,\sigma'), v \models \edgelabel{(\sphase_i,\sphase')}$ for some $\sphase' \in \autostate$, which is a contradiction.
\end{proof}
}

\para{Diagnosis of Abstract Counterexample Traces}
When a user is presented with an abstract counterexample trace, diagnosing the cause of the trace assists the user in understanding whether (i) the program is faulty, (ii) the phase structure needs to be modified (and how), or (iii) the language of phase characterizations is not expressive enough to capture the required characterizations. These cases can be differentiated by performing bounded model checking along the given abstract counterexample trace.
If a concrete trace is found, it demonstrates whether the system is not safe or the automaton needs to be changed. Otherwise, the counterexample is attributed to the limited logical language used for candidate characterizations. In this case, the user can proceed by extending the logical language, or modify the program and/or automaton so that they admit an inductive phase invariant in the given logical language.

\commentout{
\subsubsection{Diagnosis of Abstract Counterexample Traces.}
As we have shown, an abstract counterexample trace is evidence that the given phase automaton cannot be strengthened to a safe inductive phase invariant.
Diagnosing the cause of the trace assists the user in understanding whether (i) the program is faulty, (ii) the phase automaton needs to be modified, or (iii) the language of phase characterizations is not expressive enough to capture the required strengthening:. We elaborate on these possibilities below.

The abstract counterexample trace may represent a concrete trace. If the trace violates the original safety property (enforced by \Cref{eq:chc-safe}), then the transition system at hand is not safe and no safe inductive strengthening can be expected, for any phase automaton (by \Cref{cor:ind-safe-implies-safe}).

However, the trace may also violate one of the ``additional'' safety properties specified by the phase automaton (\Cref{eq:chc-edge-cover,eq:chc-strengthen}).
This does not indicate that the transition system is not safe, but rather that the choices of the phase automaton are intrinsically incorrect: it includes a trace such that a corresponding trace of the transition system reaches a state that violates the phase characterization, or allows a transition that is not covered by any outgoing edge in the automaton. If the latter happens for a \emph{deterministic} phase automaton (where every trace of $\TS$ corresponds to at most one trace of $\phaseauto$), then the automaton excludes true traces of the transition system and is therefore not a phase invariant, hence no inductive strengthening exists (by \Cref{lem:ind-strengthening-to-phase-inv},
an inductive strengthening is only possible when the phase automaton is a phase invariant), even though $\TS$ may be safe.
Therefore, the user has to modify the phase automaton itself.

Finally, the abstract trace may be attributed to the limited logical language used for candidate characterizations.
In this case, even if the given phase automaton $\phaseauto$ is a phase invariant, no inductive strengthening may exist.

These cases can be differentiated by performing bounded model checking along the given abstract counterexample trace. If no concrete trace is found, the user can proceed by extending the logical language of characterizations, or modify the program and/or automaton so that they admit an inductive phase invariant in the given logical language.
}

 \section{Inductive Invariant for Sharded Key-Value Store}
\label{sec:kv-ind}
\begin{figure}[H]
  \begin{lstlisting}[numbers=left, numberstyle=\tiny, numbersep=5pt, name=lockserv, morekeywords={state}]
  invariant $\forall k,n_1,n_2,v_1,v_2.$ table($n_1$,$k$,$v_1$)$\land$table($n_2$,$k$,$v_2$)$\rightarrow n_1 = n_2 \land v_1 = v_2$
  
  invariant $\forall k,n_1,n_2.$ owner($n_1$,$k$)$\land$owner($n_2$,$k$)$\rightarrow n_1=n_2$
  invariant $\forall k,n,v.$ table($n$,$k$,$v$)$\rightarrow$owner($n$,$k$)

  invariant $\forall k,\vsrc,\vdst,v,s,n.$ $\neg$(transfer_msg($\vsrc$,$\vdst$,$k$,$v$,$s$)$\land$$\neg$seqnum_recvd($\vdst$,$\vsrc$,$s$)$\land$owner($n$,$k$))
  invariant $\forall k,\vsrc,\vdst,v,s,n.$ $\neg$(unacked($\vsrc$,$\vdst$,$k$,$v$,$s$)$\land$$\neg$seqnum_recvd($\vdst$,$\vsrc$,$s$)$\land$owner($n$,$k$))

  invariant $\forall k,\vsrc_1,\vsrc_2,\vdst_1,\vdst_2,v_1,v_2,s_1,s_2.$ transfer_msg($\vsrc_1$,$\vdst_1$,$k$,$v_1$,$s_1$)$\land$$\neg$seqnum_recvd($\vdst_1$,$\vsrc_1$,$s_1$)
    $\land$transfer_msg($\vsrc_2$,$\vdst_2$,$k$,$v_2$,$s_2$)$\land$$\neg$seqnum_recvd($\vdst_2$,$\vsrc_2$,$s_2$)$\rightarrow \vsrc_1 = \vsrc_2 \land \vdst_1 = \vdst_2 \land v_1=v_2 \land s_1=s_2$
  invariant $\forall k,\vsrc_1,\vsrc_2,\vdst_1,\vdst_2,v_1,v_2,s_1,s_2.$transfer_msg($\vsrc_1$,$\vdst_1$,$k$,$v_1$,$s_1$)$\land$$\neg$seqnum_recvd($\vdst_1$,$\vdst_1$,$s_1$)
    $\land$unacked($\vsrc_2$,$\vdst_2$,$k$,$v_2$,$s_2$)$\land$$\neg$seqnum_recvd($\vdst_2$,$\vsrc_2$,$s_2$)$\rightarrow \vsrc_1 = \vsrc_2 \land \vdst_1 \land \vdst_2 \land v_1 = v_2 \land s_1 = s_2$
  invariant $\forall \vsrc_1,\vsrc_2,\vdst_1,\vdst_2,v_1,v_2,s_1,s_2.$ unacked($\vsrc_1$,$\vdst_1$,$k$,$v_1$,$s_1$)$\land$$\neg$seqnum_recvd($\vdst_1$,$\vsrc_1$,$s_1$)
    $\land$unacked($\vsrc_2$,$\vdst_2$,$k$,$v_2$,$s_2$)$\land$$\neg$seqnum_recvd($\vdst_2$,$\vsrc_2$,$s_2$)$\rightarrow \vsrc_1 = \vsrc_2 \land \vdst_1 = \vdst_2 \land v_1 = v_2 \land s_1 = s_2$
\end{lstlisting}
\caption{\footnotesize Inductive invariant for the running example.
}
\end{figure}  \section{Overview of Phase-$\UPDR$}
\label{sec:algorithm-long}
Our procedure is based on $\UPDR$~\cite{DBLP:journals/jacm/KarbyshevBIRS17}, a variant of PDR~\cite{ic3,pdr} that infers universally quantified inductive invariants.
PDR computes a sequence of \emph{frames}, $\Frame_0, \ldots, \Frame_n$
such that $\Frame_i$ overapproximates the set of states reachable in $i$ steps.
In our case, each frame $\Frame_i$ is a mapping from a phase $\sphase$ to a characterization.

We describe the gist of the procedure using the terminology of phase automata.
An inductive trace is a sequence of frames such that
$\forall \qoset. \ \Frame_{i}(\sphase) \implies \Frame_{i+1}(\sphase)$ for all $i$ and $\sphase \in \autostate$,
the first frame satisfies $\Frame_0(\autoinit) = \Init$ where $\autoinit$ is the initial phase and $\Frame_0(\sphase) = \false$ for other phases (in accordance with \Cref{eq:chc-init}),
all frames satisfy \Cref{eq:chc-edge-cover,eq:chc-safe}, and the constraint of \Cref{eq:chc-ind} is interpreted between successive frames, namely for all $0 \leq i < n$ and for all $(\sphase,\tphase) \in \autotrans$,
\begin{equation}
\forall \qoset. \ \left(\Frame_{i}(\sphase) \land \edgelabel{(\sphase,\tphase)} \implies \left(\Frame_{i+1}(\sphase)\right)'\right).
\end{equation}
These properties ensure that $\Frame_i$ induces phase characterizations such that the language of the induced phase automaton includes all traces of $\TS$ of length at most $i$.

The procedure gradually constructs the inductive trace by generating and blocking proof obligations, where a proof obligation $(m, \sphase,i)$ consists of system state(s) $m$ that need to be proved unreachable at phase $\sphase$ of the automaton at frame $i$ (i.e., with traces of length bounded by $i$).
When a proof obligation is shown to hold, it is generalized into a new lemma that excludes the corresponding states and is added as a conjunct to the characterization of phase $\sphase$ at frame $i$, where the phase characterizations of each frame are initially set to $\true$.\footnote{Our description here omits some details, such as pushing lemmas between frames; see~\cite{pdr} for further discussion.}

Proof obligations are generated by a backward traversal. First, whenever a new frame $\Frame_{n}$ is added, proof obligations are generated from counterexamples to the safety properties of \Cref{eq:chc-edge-cover,eq:chc-safe} in some phase $\sphase'$ based on $\Frame_n(\sphase')$.
Then, to check whether a proof obligation $\psi$ at phase $\sphase'$ can be blocked in frame $\Frame_{i}$, our procedure checks whether the phase characterizations of its pre-phases in the previous frame suffice to show that $\psi$ holds in $\sphase'$ in accordance with constraint (\ref{eq:chc-ind}), i.e.\ $\forall \qoset. \ \left(\Frame_{i-1}(\sphase) \land \edgelabel{(\sphase,\sphase')} \implies \psi'\right)$ for all $\sphase$ such that $(\sphase,\sphase') \in \autotrans$. Otherwise, there is a pre-phase $\sphase$, a valuation $v$ and a transition $(\sigma,\sigma'), v \models \edgelabel{(\sphase,\sphase')}$ such that $\sigma, v \models \Frame_{i-1}(\sphase)$ but $\sigma', v \not\models \psi$. This generates a proof obligation $\theta$ for phase $\sphase$ in frame $i-1$.

In $\UPDR$, the proof obligation $\theta$ is constructed as the diagram of the counterexample, which is the strongest existentially quantified abstraction of $\sigma$~\cite{chang1990model}. In our case, $\theta = \diag{\sigma, v, \qoset}$ is defined by
\begin{equation}
	\diag{\sigma, v, \qoset} = \exists x_1,\ldots,x_{m}. \
			\psi_{\text{distinct}} \land \psi_{\text{identity}} \land \psi_{\text{rels}}
\end{equation}
where
\begin{itemize}
    \item $\sigma$ has domain $\{e_1,\ldots,e_m\}$ and interpretation $\Int$.
	\item $\psi_{\text{distinct}}$ is a conjunction of all inequalities $x_i \neq x_j$ for $i \neq j$.

	\item $\psi_{\text{identity}}$ is a conjunction of the equalities $x_i = c$ if $c$ is a constant symbol and $\Int(c)=e_i$, and of the equalities $x_i = y$ if $y \in \qoset$ and $v(y) = e_i$.

	\item $\psi_{\text{rels}}$ is a conjunction of atomic formulas $r(x_{i_1}, \ldots, x_{i_a})$ for every relation symbol $r$ of arity $a$ and elements $\ov{e} = e_{i_1},\ldots,e_{i_a}$ such that $\ov{e} \in \Int(r)$, and $\neg r(x_{i_1}, \ldots, x_{i_a})$ if $\ov{e} \not\in \Int(r)$. Function symbols are treated similarly.
\end{itemize}	

Note that $\diag{\sigma,v,\qoset}$ has $\qoset$ as free variables.
A proof obligation $(\diag{\sigma,v,\qoset}, \sphase, i)$ is blocked by adding to the characterization of phase $\sphase$ at frame $i$ a universally quantified clause (possibly with free variables in $\qoset$) that implies $\neg \diag{\sigma,v,\qoset}$.

The procedure continues to generate (and block) proof obligations across automaton edges, going backwards in the frames, until it reaches the initial frame.
If a proof obligation does not hold there, the procedure has found an abstract counter-trace, and returns it as evidence of absence of a safe inductive phase invariant. The reason is that the backward traversal is performed over diagrams, where $({\sigma},v) \models \diag{\tilde{\sigma},v,\qoset}$ if and only if $\tilde{\sigma}$ is a substructure of ${\sigma}$ (up to isomorphism)~\cite{chang1990model}, and hence if and only if $(\tilde{\sigma},v) \leqU ({\sigma},v)$ (\Cref{ex:subtructure}).
Otherwise, the procedure terminates if one of the frames constitutes a solution to the CHC system, namely an inductive phase invariant. \section{Proofs}
\label{sec:proofs}

\begin{proof*}{Proof of \Cref{lemma:phase-ind-to-phase-inv}}
Given a domain $\Dom$ and a valuation $v$, the phase characterizations of $\phaseauto$ induce the following simulation relation from $\TS$ to $\phaseauto$:
$H = \{(\sigma, \sphase) \, | \, \sigma, v \models \statelabel{\sphase}\}$
in the sense that:
\begin{enumerate}
	\item (\emph{Labeling}) $(\sigma, \sphase) \in H$ implies $\sigma,v \models \statelabel{\sphase}$.

	\item (\emph{Initial}) For every $\sigma_0 \models \Init$, $(\sigma_0, \autoinit) \in H$.
		  This follows from the initiation property of $\phaseauto$ (\Cref{def:ind}).

	\item (\emph{Step}) For every transition $(\sigma, \sigma') \models \TR$, if $(\sigma, \sphase) \in H$ then there exists $\tphase \in \autostate$ such that  $(\sigma', \tphase) \in H$ and
	$(\sigma, \sigma'), v \models \edgelabel{(\sphase,\tphase)}$.
	This holds due to the edge covering and inductiveness properties of $\phaseauto$ (\Cref{def:ind}):
	From edge covering, there exists $\tphase \in \autostate$ such that $(\sigma, \sigma'), v \models \edgelabel{(\sphase,\tphase)}$.
	From inductiveness, necessarily $\sigma', v \models \statelabel{\tphase}'$. 
	Thus $(\sigma', \tphase) \in H$, as required.
\end{enumerate}
Hence, for every $\Dom,v$, trace inclusion follows.
\qed
\end{proof*}

\begin{proof*}{Proof of \Cref{lemma:phase-to-ind}}
	Inductive phase invariant to inductive invariant:
	$\Init \implies \text{Inv}_{\phaseauto}$ follows from the initiation condition of $\phaseauto$.
	For consecution, assume $\sigma \models \text{Inv}_{\phaseauto}$ and $(\sigma, \sigma') \models \TR$.
	Then there is some $\sphase \in \autostate$ such that $\sigma \models \statelabel{\sphase}$.
	From edge covering condition of $\phaseauto$, there exists $\tphase \in \autostate$ such that $(\sigma, \sigma') \models \edgelabel{(\sphase,\tphase)}$.
	From the inductiveness condition of $\phaseauto$ necessarily $\sigma' \models \statelabel{\tphase}'$, and thus $\sigma' \models \text{Inv}'_{\phaseauto}$, as required.
	The converse direction is clear.
	\qed
\end{proof*}

\begin{proof*}{Proof of \Cref{cor:ind-safe-implies-safe}}
Consider the inductive invariant $I = \forall \qoset. \ \bigvee_{\sphase \in \autostate}{\statelabel{\sphase}}$ per \Cref{lemma:phase-to-ind}. The safety of $\phaseauto$ (\Cref{def:safe-automaton}) implies that $I \implies \forall \qoset. \, \Safety$, and the claim follows.

It is also possible to prove the lemma through phase invariants, by claiming that if $\phaseauto$ is safe and an invariant for $\TS$ then $\TS$ is safe:
Let $\sigma_1,\ldots,\sigma_n$ be a finite trace of $\TS$.
	Let $v$ be a valuation for $\qoset$.
	$\ov{\sigma} \models \phaseauto$, so there exists a trace of phases $\sphase_0,\ldots,\sphase_n$ such that $\sigma_i, v \models \statelabel{\sphase_i}$.
	Since $v \models \statelabel{\sphase_i} \to \Safety$, for every $i$ it holds that $\sigma_i, v \models \Safety$.
The claim follows.
	\qed
\end{proof*}

\begin{proof*}{Proof of \Cref{lem:inference}}
Assume that $\mathbf{I}$ is a solution to the CHC system. Then $\phaseauto$ is a safe inductive phase invariant for $\TS$: $\phaseauto$ satisfies initiation due to constraint \ref{eq:chc-init}, satisfies inductiveness due to constraint \ref{eq:chc-ind}, and edge covering due to constraint \ref{eq:chc-edge-cover}, and thus $\phaseauto$ is inductive w.r.t.\ $\TS$ (\Cref{def:ind}). Finally, $\phaseauto$ is safe (\Cref{def:safe-automaton}) due to constraint \ref{eq:chc-safe}.

Conversely, assume that $\phaseauto$ is an inductive phase invariant. Then constraint \ref{eq:chc-init} is satisfied because $\phaseauto$ satisfies initiation, constraint \ref{eq:chc-ind} is satisfied because $\phaseauto$ satisfies inductiveness, and constraint \ref{eq:chc-edge-cover} because $\phaseauto$ satisfies edge covering---all from the definition of an inductive phase invariant (\Cref{def:ind}). Finally, constraint \ref{eq:chc-safe} is satisfied because $\phaseauto$ is safe (\Cref{def:safe-automaton}).
\qed
\end{proof*}

\begin{proof*}{Proof of \Cref{lem:strengthening-exists}}
The implication from right to left is a consequence of \Cref{lemma:phase-ind-to-phase-inv}. 
Consider the other direction. We construct an inductive phase invariant as follows:
For a given valuation $v$, characterize each phase by the set of states that can reach this phase: reachability of $\sigma$ to phase $\sphase$ means that there exist trace of program states $\sigma_0,\ldots,\sigma_n=\sigma$ and a matching trace of phases $\sphase_0,\ldots,\sphase_n=\sphase$ per \Cref{def:language-phase-automaton}.
The result is indeed an inductive phase invariant: initiation follows from $\phaseauto$ being a phase invariant.
Inductiveness follows from taking the characterizations to be sets of reachable states (recall that the language of characterizations is assumed to be unrestricted).
It remains to argue that edge covering holds. Let $\sigma$ be reachable in phase $\sphase$ and $(\sigma,\sigma') \models \TR$. Thus there is a sequence of states program states $\sigma_0,\ldots,\sigma_n=\sigma$ and a matching trace of phases $\sphase_0,\ldots,\sphase_n=\sphase$.
Now, $\sigma'$ is also reachable, and since $\phaseauto$ is a phase invariant there exists a trace of phases trace of phases $\sphase'_0,\ldots,\sphase'_{n},\sphase'_{n+1}=\sphase'$ matching $\sigma_0,\ldots,\sigma_{n},\sigma'$.
But $\phaseauto$ is \emph{deterministic}, so necessarily $\sphase'_{n} = \sphase_{n}$ (by induction over $n$) which is $\sphase$. In particular this gives $(\sigma,\sigma') \models \edgelabel{(\sphase,\tphase)}$ where $\tphase = \sphase'_{n+1}$, as required.
\OMIT{
	For a deterministic phase automaton, $\Lang{\TS} \subseteq \Lang{\phaseauto}$ if and only if there exists a simulation relation between $\TS$ and $\phaseauto$.
Furthermore, in this case, the inductiveness requirement coincides with the requirement that the phase characterizations induce a simulation relation. Hence, in this case, by defining the characterization of $\sphase$ to include all the states simulated by it, we obtain an inductive strengthening of $\phaseauto$.	
}
\qed
\end{proof*}

\begin{proof*}{Proof of \Cref{lem:determinization}}
The definition of $\edgelabel{}'$ ensures that it is deterministic, and inherits the edge covering property from $\edgelabel{}$.
Initiation is not affected by the edge labeling, and inductiveness cannot be damaged by strengthening $\edgelabel{(\sphase,\tphase)}'$.
\qed
\end{proof*}

\begin{proof*}{Proof of \Cref{lem:cex-trace}}
Follows from \Cref{lem:abstract-cex-trace} with the identity preorder (or simply by reiterating the proof while ignoring abstraction steps).
\end{proof*}

\begin{proof*}{Proof of \Cref{lem:abstract-cex-trace}}
    For the sake of the proof, we explicitly split each transition in an abstract trace to an abstraction step followed by a transition.
Let $(\sigma_1,\tilde{\sigma}_1,\ldots,\sigma_n,\tilde{\sigma}_n),v)$ be such a trace and $\sphase_1,\ldots,\sphase_n$ be as in the definition.
    Let $\phaseauto'$ be a strengthening of $\phaseauto$ with phase characterization $\statelabel{}'$.
	Assume that $\phaseauto'$ is safe and show that $\phaseauto'$ is not inductive.
	Assume for the sake of contradiction that it is.
	We claim by induction on $i$ that $\sigma_i, v \models \statelabel{\sphase_i}'$.
	The base claim follows from initiation.
	For the induction step, assume that $\sigma_i, v \models \statelabel{\sphase_i}'$.
	Since $(\tilde{\sigma}_i, v) \leqL (\sigma_i, v)$ and $\statelabel{\sphase_i}' \in \Lset$, $\tilde{\sigma}_i, v \models \statelabel{\sphase_i}'$.
	Now $(\tilde{\sigma}_i, \sigma_{i+1}), v \models \edgelabel{(\sphase_i, \sphase_{i+1})}$, and from the assumption that $\phaseauto'$ is inductive necessarily $\sigma_{i+1}, v \models \statelabel{\sphase_{i+1}}'$, as required.

	Let us consider the cause of the abstract counterexample trace. In case \ref{it:abs-cex-unsafe}, $\sigma_i, v \not \models \Safety$ is a contradiction to the safety of $\phaseauto'$ (\Cref{def:safe-automaton}).
	\ifsketch
	In case \ref{it:abs-cex-unchar}, $\sigma_i, v \not \models \statelabel{\sphase_i}$ implies that $\sigma_i, v \not \models \statelabel{\sphase_i}'$ (strengthening), which we have shown to be impossible.
	\fi
	In case \ref{it:abs-cex-untrans}, if $(\sigma_i, \sigma') \models \TR$, since we have shown $\sigma_i, v \models \statelabel{\sphase_i}'$ it must follow from edge covering that
    $(\sigma_i,\sigma'), v \models \edgelabel{(\sphase_i,\sphase')}$ for some $\sphase' \in \autostate$, which is a contradiction.
    The claim follows.
\qed
\end{proof*} \fi

\end{document}